\documentclass[12pt, draftclsnofoot, onecolumn]{IEEEtran}
\usepackage{epsfig,latexsym}
\usepackage{float}
\usepackage{indentfirst}
\usepackage{amsmath}
\usepackage{amssymb}
\usepackage{times}
\usepackage{subfigure}
\usepackage{psfrag}
\usepackage{hyperref}
\usepackage{cite}
\usepackage{lastpage}
\usepackage{fancyhdr}
\usepackage{color}
 \usepackage{amsthm}
\usepackage{bigints}
\newtheorem{theorem}{Theorem}
\newtheorem{Corollary}{Corollary}

\newtheorem{corollary}[Corollary]{$\mathbf{Corollary}$}

\newcounter{problem}
\newcounter{save@equation}
\newcounter{save@problem}
\makeatletter

\begin{document}
\title{ {  A New QoS-Guarantee     Strategy for NOMA Assisted Semi-Grant-Free Transmission  }}

\author{ Zhiguo Ding, \IEEEmembership{Fellow, IEEE}, Robert Schober, \IEEEmembership{Fellow, IEEE}, and H. Vincent Poor, \IEEEmembership{Life Fellow, IEEE}    \thanks{ 
  
\vspace{-2em}

    Z. Ding and H. V. Poor are  with the Department of
Electrical Engineering, Princeton University, Princeton, NJ 08544,
USA. Z. Ding
 is also  with the School of
Electrical and Electronic Engineering, the University of Manchester, Manchester, UK (email: \href{mailto:zhiguo.ding@manchester.ac.uk}{zhiguo.ding@manchester.ac.uk}, \href{mailto:poor@princeton.edu}{poor@princeton.edu}).
R. Schober is with the Institute for Digital Communications,
Friedrich-Alexander-University Erlangen-Nurnberg (FAU), Germany (email: \href{mailto:robert.schober@fau.de}{robert.schober@fau.de}).

  }\vspace{-3em}}
 \maketitle
 
\begin{abstract}  
Semi-grant-free (SGF) transmission  has recently received significant attention due to its capability to accommodate massive connectivity and reduce access delay by admitting grant-free users to   channels which would otherwise be solely occupied by grant-based users. In this paper, a new    SGF transmission scheme that exploits the flexibility in choosing the  decoding order  in non-orthogonal multiple access (NOMA) is proposed. Compared to existing SGF schemes, this new scheme can ensure that admitting the grant-free users is  completely transparent to the grant-based users, i.e., the grant-based users' quality-of-service experience is guaranteed to be the same as for  orthogonal multiple access. In addition, compared to existing SGF schemes,  the proposed SGF scheme can   significantly improve    the robustness of the grant-free users' transmissions and effectively avoid outage probability    error floors. To facilitate the performance evaluation of the proposed SGF transmission scheme, an exact expression for the outage probability   is obtained and an asymptotic analysis  is conducted to show that the achievable multi-user diversity gain is proportional to the number of participating   grant-free users. Computer simulation results  demonstrate the performance of the proposed SGF transmission scheme and   verify the accuracy of the developed analytical results.  
\end{abstract} \vspace{-1em}

\section{Introduction}
The next generation Internet of Things (NGIoT) is envisioned to be an important use case for beyond 5G mobile networks \cite{ngiot}. The key challenge for supporting NGIoT is, given the scarce radio spectrum, how to support a massive number of devices, each of which     might   send a small number of packets only. For this emerging application, conventional grant-based transmission is not suitable, since the amount of signalling needed for handshaking could exceed   the amount of data sent by the devices.  This motivates the development  of  grant-free transmission schemes, which grant  the devices  access without   lengthy handshaking protocols \cite{6933472}. Most existing grant-free schemes  can be categorized  into three groups. The first group applies  random access protocols originally developed for computer networks \cite{randomaccess1}, the second  group   relies on the excess spatial degrees of freedom offered by multiple-input multiple-output (MIMO) techniques \cite{8454392,8734871}, and the third group employs  non-orthogonal multiple access (NOMA) which encourages spectrum sharing among the devices   \cite{mojobabook,8419284,8703780,7976275,8533378}. It is noted  that  there are many works which use  a combination of the three types of grant-free schemes  and hence can potentially offer a significant performance improvement in terms of connectivity and transmission robustness \cite{8674774,8719976,jsacnoma10}. 

In this paper, we focus on a special case of NOMA based grant-free transmission, termed semi-grant-free (SGF) transmission \cite{8662677}. Unlike the aforementioned pure grant-free schemes, SGF transmission does not assume that a certain number of resource blocks, such as time slots or subcarriers, are reserved for contention among the grant-free users, since this assumption would put a strict cap on the number of grant-free users which can be served, particularly if the base station has a limited number of antennas and cannot use massive MIMO to improve connectivity.  The key idea of SGF transmission is to opportunistically admit grant-free users to those resource blocks  which would otherwise be solely occupied by grant-based users. An immediate advantage of SGF over conventional grant-free schemes is that the number of grant-free users is constrained not by the number of resource blocks reserved for  grant-free transmission, but by the total number of resource blocks available in the system.   Take an orthogonal frequency-division multiple access (OFDMA) system with $128$ subcarriers as an example. If  only $8$ subcarriers are reserved for grant-free transmission, at most $8$ grant-free users can be served, but the use of SGF transmission  can potentially provide service to     $120$ additional  grant-free users. 

In SGF transmission, a crucial task is how  to guarantee a grant-based user's quality of service   (QoS) experience when    admitting grant-free users to the same resource block.  
In \cite{8662677}, two SGF transmission schemes, termed SGF Scheme I and Scheme II, were developed to realize this goal.   In particular, SGF Scheme I  requires  the base station to    decode the grant-based user's signal  first by treating the grant-free users' signals as interference, and schedules   grant-free users with weak channel conditions   in order to limit the interference they cause to the grant-based user.   Therefore,  Scheme I is ideal for  situations, where the grant-free users are cell-edge users, i.e., their connections to the base station are weaker than that of the grant-based user. SGF Scheme II schedules grant-free users with strong channel conditions, and requires the base station to  decode the grant-free users' signals first. Therefore, Scheme II is ideal for   situations, where the grant-free users' connections to the base station are strong.  Two types of distributed contention control were applied in \cite{8662677} to reduce the system overhead and to control the number of admitted grant-free users \cite{Zhao2005s, Bletsas06, 6334506}. 

In this paper, we consider the same   grant-free communication scenario as in \cite{8662677}, i.e.,   one grant-based user and $M$ grant-free users communicate in one resource block with the same base station. A new SGF transmission scheme is proposed which can be interpreted   as  an opportunistic  combination  of the two existing SGF schemes and offers the following three advantages:
\begin{itemize}
\item   Recall that SGF Scheme I decodes the grant-based user's signal first  by treating the grant-free users' signals as interference. Hence, it is inevitable that the grant-based user's QoS experience  is negatively affected by the admission of the grant-free users into the channel.  The new SGF scheme can strictly guarantee  that  admitting grant-free users is transparent to the grant-based user, and the grant-based user's QoS experience is the same as   when it   occupies   the channel along.
 
 \item Recall that SGF Scheme II   directly decodes the grant-free users' signals by treating the grant-based user's signal as interference, which means that interference always exists for the grant-free users. Hence, for SGF Scheme II,  the data rates available for the grant-free users  can be small. The new SGF scheme can realize interference-free transmission for the grant-free users,   and hence can offer  significantly improved   achievable data rates for grant-free transmission. 
 
 \item For both existing SGF schemes, their outage probabilities exhibit   error floors, when there is no   transmit power control, e.g.,   the grant-free and grant-based users increase their transmit powers   without coordination. Take Scheme I as an example, which  decodes the grant-based user's signal first and then decodes the grant-free users' signals via successive interference cancellation (SIC).  An     outage probability error floor exists because increasing the grant-free users' transmit powers might help the second stage of SIC but increases the outage probability in the first stage. A similar error floor exists for Scheme II. The new SGF scheme can effectively avoid these    error floors and significantly improve transmission robustness, even without careful power control among the users.  
 
 \end{itemize}
 
 In order to facilitate the performance analysis, an exact expression for the outage probability achieved by the proposed SGF transmission scheme is developed based on  order statistics. Because the outage probability achieved by the proposed SGF scheme can be a function of four random variables, including three dependent order statistics,  the developed exact expression has an involved form and hence   cannot provide much   insight into the properties of  SGF transmission.  Therefore, two high SNR approximations are developed based on an asymptotic analysis of  the derived  exact expression. The asymptotic expressions   demonstrate  that the proposed SGF transmission scheme   avoids an outage probability error floor   and realizes a multi-user diversity gain of $M$. 

The remainder  of the paper is organized as follows. In Section \ref{section II}, the existing SGF schemes are briefly introduced first, and then, the proposed new SGF   scheme is described. In Section \ref{section III}, the outage performance achieved by the proposed SGF transmission scheme is analyzed, where   an asymptotic analysis is also conducted to illustrate the multi-user diversity gain realized by the proposed scheme.   Computer simulations   are provided in Section \ref{section IV}, and   Section \ref{section V} concludes the paper. We collect the details of all  proofs   in the appendix.

\section{Existing and Newly Proposed   SGF Schemes}  \label{section II}
Consider an  SGF communication scenario, where $M$ grant-free users compete with each other for admission  to a resource block which would otherwise  be solely occupied by a grant-based user for conventional orthogonal multiple access (OMA).  Denote the grant-based user's channel gain  by $g$, and the grant-free users' channel gains by $h_m$, $1\leq m \leq M$. We assume that the SGF system  operates in quasi-static Rayleigh fading environments, i.e., all the channel gains are  complex Gaussian distributed with zero mean and unit variance. Without loss of generality, we also assume that the grant-free users' channel gains are ordered as follows:
\begin{align}\label{channel order}
|h_1|^2\leq \cdots\leq |h_M|^2.
\end{align}
We note that this ordering assumption is to facilitate the performance analysis, and that this information is not available to any of the nodes in the system, including the base station.
Prior to transmission, we assume that the grant-free users can overhear the information exchange between the grant-based user and the base station, and hence know the grant-based user's channel state information (CSI) as well as the grant-based user's transmit power, denoted by $P_0$. In addition, each grant-free user has acquired the knowledge of  its own CSI, by exploiting  the pilot signals broadcasted by the base station.

\subsection{ Two Existing SGF Schemes}
SGF Scheme I in \cite{8662677} requires  the base station to decode the grant-based user's signal during the first stage of SIC. If grant-free user $m$ is admitted to the channel, the grant-based user's achievable data rate is  $\log\left(1+\frac{P_0|g|^2}{P_s|h_m|^2+1}\right)$, and grant-free user $m$'s data rate is $\log(1+P_s|h_m|^2)$ if the first stage of SIC is successful, where   $P_s$ denotes the transmit power of the    grant-free users,  and the noise power is assumed to be normalized to one.  In order to guarantee   the grant-based user's QoS requirement,  the base station broadcasts   a predefined threshold, denoted by $\tau_I$, and only the grant-free users whose channel gains fall below the threshold participate in contention. In this way,   a user which has a strong channel   and hence can cause strong interference to the grant-based user will not be granted access.

The use of distributed contention control ensures that contention can be carried out in a distributed manner. Thereby,  each grant-free user's backoff period is proportional to its channel gain and therefore the user with the weakest channel gain will be granted access\footnote{In this paper, we focus on the case, where  a single grant-free user is admitted to the channel. However, as discussed in \cite{8662677} and \cite{6334506}, more than one user can be granted access via distributed contention, which is beyond the scope of this paper.  }.  Therefore, for SGF Scheme I, the admitted grant-free user's data rate    is given by
\begin{align}
R_{I} = \left\{\begin{array}{ll} \log(1+P_s|h_1|^2), &{\rm if} \text{ } |h_1|^2\leq \tau_I\quad \& \quad \log\left(1+\frac{P_0|g|^2}{P_s|h_1|^2+1}\right)>R_0\\ 0, &{\rm otherwise}\end{array}\right.,
\end{align}
where $R_0$ denotes the grant-based user's target data rate.

SGF Scheme II in \cite{8662677} requires  the base station to decode a grant-free user's signal during the first stage of SIC. Similar to Scheme I, the base station broadcasts a threshold, denoted by $\tau_{II}$, and only the grant-free users whose channel gains are stronger than the threshold participate in contention. By using distributed contention control, each participating grant-free user sets its backoff time inversely propotional to its channel gain, which means that the grant-free user with the strongest channel condition is granted access, and its   achievable data rate is given by
\begin{align}
R_{II} = \left\{\begin{array}{ll} \log\left(1+\frac{P_s|h_M|^2}{P_0|g|^2+1}\right), &{\rm if} \text{ } |h_M|^2\geq \tau_{II} \\ 0, &{\rm otherwise}\end{array}\right.. 
\end{align}
 
{\it Remark 1:} We note that,  when $P_s\rightarrow \infty$ and $P_0\rightarrow \infty$, there is an error floor for the admitted grant-free user's outage probability. Take SGF Scheme II as an example. $\log\left(1+\frac{P_s|h_M|^2}{P_0|g|^2+1}\right)$ becomes a constant, when $P_s\rightarrow \infty$ and $P_0\rightarrow \infty$, which means that there will be an error floor for the outage probability suffered  by Scheme II. This error floor can be reduced, if $P_s$ is much larger than $P_0$.  In other words, the   existing SGF schemes require careful power control to guarantee  the grant-free user's target outage performance, which might not be possible in practice.  

\subsection{Proposed   SGF Scheme} 
  In the proposed SGF scheme, prior to transmission, the base station broadcasts a threshold, denoted by $\tau(|g|^2)$, which needs to ensure the following inequality:
 \begin{align}
\label{tau1}
\log\left(1+\frac{P_0|g|^2}{\tau(|g|^2)+1}\right)\geq R_0.
\end{align}
The proposed SGF scheme  chooses   $\tau(|g|^2)$ such that the above  inequality constraint holds with equality: 
\begin{align}
\tau(|g|^2) =\max\left\{0, \frac{P_0|g|^2}{2^{R_0}-1}-1\right\},
\end{align}
where $\max (a,b)$ denotes the maximum of $a$ and $b$.

 Upon receiving this threshold, each grant-free user compares its channel gain with the threshold individually. Unlike the two existing schemes, the proposed SGF scheme allows all the grant-free users to participate in contention. Each user's backoff time is determined by how its channel gain compares  to $\tau(|g|^2)$, as shown in the following:
  
  \begin{itemize}  
  \item Group $1$  contains the users  whose   channel gains are  above the threshold, i.e., $P_s|h_m|^2>\tau(|g|^2)$. If a user in Group $1$ is granted access, its signal has to be detected during the first stage of SIC. Otherwise, $P_s|h_m|^2>\tau(|g|^2)$ leads to $\log\left(1+\frac{P_0|g|^2}{P_s|h_m|^2+1}\right)<R_0$, which would mean  that the grant-based user's signal cannot be decoded correctly \footnote{In this paper, the grant-free users are assumed to use the same   fixed transmit power, $P_s$. The use of distributed power control can further improve the performance of SGF transmission without increasing system overhead, but is beyond the scope of this paper.    }.   Therefore, if a user in Group $1$ is granted access, its achievable data rate is $\log\left(1+\frac{P_s|h_m|^2}{P_0|g|^2+1}\right) $, and hence   its backoff time is set to be inversely proportional to this achievable rate.  
  
    \item Group $2$ contains the users  whose   channel gains are below the threshold, i.e., $P_s|h_m|^2<\tau(|g|^2)$. For a user in Group $2$, its signal can be decoded in either one of the two SIC stages,  without affecting the grant-based user's QoS. In particular, if its signal is decoded in the first stage of SIC,  its achievable data rate is $ \log\left(1+\frac{P_s|h_m|^2}{P_0|g|^2+1}\right) $.  If its signal is decoded in the second stage of SIC,   its achievable data rate is $\log(1+P_s|h_m|^2)$.  We note that if a user   from Group $2$ is granted access, it is guaranteed that  the grant-based user's signal can be successfully decoded in the first stage of SIC,  since $P_s|h_m|^2<\tau(|g|^2)$ leads to $\log\left(1+\frac{P_0|g|^2}{P_s|h_m|^2+1}\right)>R_0$. In other words, $\log(1+P_s|h_m|^2)$ is always achievable for a user from Group $2$.  Therefore,  its backoff time is set to be inversely proportional to $ \log(1+P_s|h_m|^2) $, since  $ \log(1+P_s|h_m|^2)\geq  \log\left(1+\frac{P_s|h_m|^2}{P_0|g|^2+1}\right) $. 
  \end{itemize}
  
 By carrying out distributed contention control \cite{Zhao2005s, Bletsas06, 6334506}, either a user from Group 1 with the largest $\log\left(1+\frac{P_s|h_m|^2}{P_0|g|^2+1}\right) $ or a user from Group 2 with the largest $ \log(1+P_s|h_m|^2) $ is granted access in a distributed manner. 

{\it Remark 2:} The proposed  SGF scheme can be viewed as a hybrid version of the two existing schemes. In particular, under the condition that admitting a   grant-free user needs to be transparent to the grant-based user, the users in Group 1 can support SGF Scheme II, whereas the users in Group 2 can support either of the two schemes.  The proposed scheme will select the grant-free user with the largest achievable data rate in an opportunistic manner .  

{\it Remark 3: } We note that, among the $M$ grant-free users, only two users have the chance of being granted access, if the grant-free users' channel gains are ordered as in \eqref{channel order}. One   is grant-free user $M$, if Group $1$ is not empty, since $\log\left(1+\frac{P_s|h_m|^2}{P_0|g|^2+1}\right) \leq \log\left(1+\frac{P_s|h_j|^2}{P_0|g|^2+1}\right) $  always holds for any $m\leq  j$. The other one is the grant-free user which has the strongest channel gain in Group $2$,  if Group $2$ is not empty, since $ \log(1+P_s|h_m|^2) \leq  \log(1+P_s|h_j|^2) $ for $m\leq j$. If a grant-based protocol is used, i.e.,  global CSI is available at the base station, the base station can    decide which user is to be admitted by simply comparing the two users' data rates. The use of the proposed distributed contention  control can ensure that the same goal is achieved without acquiring global CSI at the base station. 
 
\section{Outage Performance Analysis} \label{section III}
It is straightforward to show that the use of the proposed SGF scheme can strictly guarantee that admitting grant-free users is completely transparent to the grant-based user, and the grant-based user's experience is the same as with OMA. Therefore, in this paper, we mainly focus on the outage performance of the admitted grant-free user, where we assume that all the grant-free users have the same target data rate, denoted by $R_s$.  

 To characterize the outage event,    denote the event that there are $m$ users in Group 2 by $E_m $, where $E_m$ can be   explicitly defined  as follows:
\begin{align}
E_m = \left\{ |h_{m}|^2<\frac{\tau(|g|^2)}{P_s},    |h_{m+1}|^2> \frac{\tau(|g|^2)}{P_s}  \right\}   ,
\end{align}
for   $1\leq m \leq M-1$. Furthermore, the two extreme cases with no user in Group 1 and Group 2 can be defined  as $E_M=\left\{ |h_{M}|^2<\frac{\tau(|g|^2) }{P_s}\right\}$ and $E_0 = \left\{ |h_{1}|^2>\frac{\tau(|g|^2)}{P_s}   \right\}   $, respectively. 

The overall outage probability experienced by the admitted grant-free users is given by 
\begin{align}\nonumber
{\rm P}_{out} = & \sum^{M-1}_{m=1}{\rm P}\left(E_m,  \max\left\{ R_{k,I} ,1\leq k\leq m \right\}<R_s,\max\left\{R_{k,II} ,m< k\leq M \right\}<R_s  \right)\\  &+ {\rm P}\left(E_M,   \max\left\{ R_{k,I} ,1\leq k\leq M \right\}<R_s\right) \label{out}
+ {\rm P}\left(E_0, \max\left\{R_{k,II},1\leq k\leq M \right\}<R_s\right),
\end{align}
where $R_{k,I} = \log\left(1+P_s|h_k|^2\right)$ and $R_{k,II}=\log\left(1+\frac{P_s|h_k|^2}{P_0|g|^2+1}\right)$.  

Because the grant-free users' channel gains are ordered as in \eqref{channel order}, the outage probability can be simplified as follows:
\begin{align} \label{out overall}
{\rm P}_{out} = & \sum^{M}_{m=1}{\rm P}\left(E_m,     \max \left\{R_{m,I},R_{M,II} \right\}<R_s\right)+ {\rm P}\left(E_0,  R_{M,II} <R_s\right). 
\end{align}

 Define $\epsilon_0=2^{R_0}-1$, $\epsilon_s=2^{R_s}-1$, $\alpha_0=\frac{\epsilon_0}{P_0}$, and $\alpha_s=\frac{\epsilon_s}{P_s}$.  
We note that if $|g|^2<\alpha_0$,
\begin{align}\label{fact1}
\tau(|g|^2) =\max\left\{0, \frac{P_0|g|^2}{2^{R_0}-1}-1\right\}=0. 
\end{align}

By using \eqref{fact1}, the outage probability can be rewritten as follows: 
\begin{align}\nonumber
{\rm P}_{out} =&\sum^{M}_{m=1}\underset{Q_m}{\underbrace{{\rm P}\left(E_m,|g|^2>\alpha_0,  \max \left\{R_{m,I}, R_{M,II}\right\} <R_s\right)}}\\  &+
 \underset{Q_0}{\underbrace{{\rm P}\left(E_0,|g|^2>\alpha_0,   R_{M,II}  <R_s\right)}}  + \underset{Q_{M+1}}{\underbrace{{\rm P}\left(R_{M,II}<R_s , |g|^2<\alpha_0\right)}}.\label{overall}
\end{align} 

The following theorem provides an exact expression for the outage probability achieved by the proposed SGF scheme. 
\begin{theorem}\label{theorem1}
Assume that $\epsilon_s\epsilon_0<1$ and $M\geq 2$. The outage probability achieved by the proposed SGF transmission scheme can be expressed as follows:
 \begin{align}\nonumber
{\rm P}_{out} =&\sum^{M-2}_{m=1}     
\bar{\eta}_m  \sum^{M-m}_{l=0}{M-m \choose l}(-1)^l  \sum^{m}_{p=0}{m \choose p}(-1)^p   \tilde{\mu}_4  \phi(p,\tilde{\mu}_2)  
  \\\nonumber &+\sum^{M-1}_{i=0}{M-1 \choose i}  \frac{(-1)^i\tilde{\eta}_0}{M-1} \left(
  e^{\frac{ 1}{P_s} } \phi(i,\mu_7)   - e^{-\alpha_s } 
\phi(i,\mu_8)
 \right)
 \\\nonumber &
 +     \frac{ \tilde{\eta}_0}{M(M-1)} \sum^{M}_{l=0}(-1)^l  {M \choose l}e^{ - l\alpha_s  }  
e^{\frac{M-l}{P_s}   } g_{\tilde{\mu}_{12}}(\alpha_0,\alpha_2) 
+\sum^{M}_{i=0}{M\choose i} (-1)^i e^{\frac{ i}{P_s}}g_{\frac{i }{\alpha_0P_s}}\left(\alpha_0,\alpha_1\right)
 \\  
 &+\left(1-e^{-  \alpha_s} \right) ^M  e^{- \alpha_1 } 
+\sum^{M}_{i=0}{M\choose i}(-1)^ie^{- i\alpha_s }
 \frac{1-e^{-\left(1+ i\alpha_sP_0 \right)\alpha_0 } }{1+ i\alpha_sP_0 }\label{overalxx},
\end{align}
 where $\bar{\eta}_m = \frac{M!}{m!(M-m)!}$, $\tilde{\eta}_0 = \frac{M!}{(M-2)!}$,  $\tilde{\mu}_2 =  l\alpha_sP_0+(M-m-l) \frac{\epsilon_0^{-1}P_0}{P_s}  $,    $\tilde{\mu}_4 = e^{-l\alpha_s+(M-m-l)\frac{ 1}{P_s} }
$,   $\mu_7=\frac{1}{P_s\alpha_0}$, $\mu_8=\alpha_sP_0$, 
 $\tilde{\mu}_{12}= l\alpha_sP_0+ (M-l)\frac{\alpha_0^{-1}}{P_s} $, $\alpha_1 =(1+\epsilon_s)\alpha_0 $, $\alpha_2=\frac{\epsilon_0(\epsilon_s+1)}{(1-\epsilon_0\epsilon_s)  P_0}$, 
$g_{\mu}(x_1, x_2) = \frac{e^{-(1+\mu) x_1}-e^{-(1+\mu) x_2}}{1+\mu}$, and $\phi(p,\mu) = e^{-p\alpha_s}g_{\mu}(\alpha_1,\alpha_2)   +e^{\frac{ p}{P_s}}  g_{\mu+\frac{ p}{P_s\alpha_0}}(\alpha_0,\alpha_1)$. 
\end{theorem}
\begin{proof}
See Appendix A.
\end{proof}
Following the steps in the proof for Theorem \ref{theorem1}, the outage probability for the case $M=1$ can be obtained straightforwardly as shown in the following corollary. 
\begin{corollary}\label{corollary0}
Assume that $\epsilon_s\epsilon_0<1$ and $M=1$. The outage probability achieved by the proposed SGF transmission scheme can be expressed as follows:
 \begin{align}
  {\rm P}_{out} =& e^{\frac{ 1}{P_s}}g_{\frac{\alpha_0^{-1}}{P_s}}\left(\alpha_0, \alpha_2\right) - e^{ -\alpha_s}  g_{\alpha_sP_0} \left(\alpha_0, \alpha_2\right) +\sum^{1}_{i=0}{1\choose i} (-1)^i e^{\frac{ i}{P_s}}g_{\frac{i }{\alpha_0P_s}}\left(\alpha_0,\alpha_1\right)
 +\left(1-e^{-  \alpha_s} \right)    e^{- \alpha_1 } .
\end{align}
\end{corollary}

{\it Remark 4:} In this paper, we mainly focus on the case $\epsilon_s\epsilon_0<1$ because  the error floor of ${\rm P}_{out}$ can be avoided in this case, i.e., the scenario with $\epsilon_s\epsilon_0<1$ is   ideal   for the application of the proposed SGF scheme.  $\epsilon_s\epsilon_0<1$ means that $R_s$ needs to be small for a given  $R_0$, which  is a realistic assumption  in practice since SGF is invoked to encourage   spectrum sharing between a grant-based user and  a grant-free user with a small target data rate. 

{\it Remark 5: } We note that for the case $\epsilon_s\epsilon_0\geq 1$, the proposed SGF scheme still works and    offers significant performance gains compared to the two existing SGF schemes, as shown in the simulation section. However, for  $\epsilon_s\epsilon_0\geq 1$,  the outage probability achieved by the proposed SGF scheme    exhibits an error floor, similar to the existing schemes. More detailed discussions will be provided in Section \ref{section IV}.

{\it Remark 6:} The outage probability expression shown in Theorem  \ref{theorem1} is  complicated,  mainly due to the fact that $Q_m$  depends on the choice of $m$. For example, for $1\leq m \leq M-2$, $Q_m$ is a function of the four   channel gains, $h_m$, $h_{m-1}$,  $h_M$, and $g$, whereas  $Q_{M-1}$ is a function of only $h_{M-1}$, $h_M$, and $g$. The fact that the $h_m$, $h_{m-1}$,  and $h_M$ are dependent order statistics   makes the expression even more involved.    However, at high SNR,   insightful approximations can be obtained as shown in the following theorem.

\begin{theorem}\label{theorem2}
Assuming that  $\epsilon_s\epsilon_0<1$, $M\geq 2$ and $P_s=P_0\rightarrow \infty$, the outage probability ${\rm P}_{out}$ can be approximated at high SNR as follows:
\begin{align}
  {\rm P}_{out} 
\approx&       \sum_{m=1}^{M-2}  
\frac{\bar{\eta}_m}{P_s^{M+1}}\epsilon_s^{m}\sum^{M-m}_{i=0} {M-m \choose i}  \left( \epsilon_s+1  \right)^{M-m-i}     
\left(  \epsilon_s- \epsilon_0^{-1}   \right)^{i} \epsilon_0^{i+1}\frac{\tilde{\alpha}_2^{i+1}-(1+\epsilon_s)^{i+1}}{i+1}  \\\nonumber &+ \sum_{m=1}^{M-2}  \frac{\tilde{\eta}_0  \left(1+\epsilon_s\right)\epsilon_s ^{M-1} 
 \epsilon_0}{P_s^{M+1}(M-1)}   \left(  (\tilde{\alpha}_2 -1-\epsilon_s )  +\frac{\epsilon_s}{M}
 \right)    \\\nonumber &+ 
\frac{\bar{\eta}_m}{P_s^{M+1} } \sum^{M-m}_{i=0} {M-m \choose i}  \left( \epsilon_s+\epsilon_0 \epsilon_s  \right)^{M-m-i}     
\left(   \epsilon_s- \epsilon_0^{-1}  \right)^{i}\epsilon_0^{i+1}\frac{\epsilon_s^{m+i+1}}{m+i+1}  
\\\nonumber &+
       \frac{ \tilde{\eta}_0}{P_s^{M+1}M(M-1)}e^{\frac{M}{P_s}   } \sum^{M}_{i=0}{M\choose i}\left(\epsilon_s + 1 \right)^{M-i}    \left( \epsilon_s -\epsilon_0^{-1} \right)^i     \epsilon_0^{i+1} \frac{\tilde{\alpha}_2^{i+1}-1}{i+1}
\\\nonumber &+\frac{1}{P_s^{M+1}(M+1)\epsilon_0^M}   \epsilon_0^{M+1}  \epsilon_s^{M+1}
+   \frac{\epsilon_s  ^M}{P_s^M}  +  \frac{\epsilon_s ^M(1+\epsilon_0)^{M+1}-1}{P_s^{M+1}(M+1)},
\end{align}
where $\tilde{\alpha}_2=\frac{ (\epsilon_s+1)}{(1-\epsilon_0\epsilon_s)  } $.
\end{theorem}
\begin{proof}
See Appendix B.
\end{proof}

{\it Remark 7:} Following the same steps as in the proof of  Theorem \ref{theorem2},  the outage probability for  the case $M=1$ can be approximated as follows:
\begin{align}
  {\rm P}_{out} 
\approx&    \frac{1}{2P_s^{2}}   \epsilon_0  \epsilon_s^{2}
+   \frac{\epsilon_s  }{P_s}   + \frac{1}{P_s^2}\left( 1 +\epsilon_s  \right)\epsilon_0 (\tilde{\alpha}_2-1).
\end{align}

By comparing the terms in Theorem \ref{theorem2}, one can find that there is one term   proportional to $\frac{1}{P_s^{M}}$, and the other ones are   proportional  to  $\frac{1}{P_s^{M+1}}$. Therefore, a further approximation can be straightforwardly obtained  as shown in the following corollary.
\begin{corollary}\label{corollary1}
Assuming that  $\epsilon_s\epsilon_0<1$ and $P_s=P_0\rightarrow \infty$, the outage probability ${\rm P}_{out}$ can be further approximated as follows:
\begin{align}
  {\rm P}_{out}  \approx  \frac{\epsilon_s  ^M}{P_s^M}   . 
\end{align}
A   diversity gain of $M$ is achievable for the proposed SGF transmission scheme. 
 \end{corollary}
 {\it Remark 8:}  Recall that for the two existing SGF schemes,    their outage probabilities suffer from error floors, when $P_s$ and $P_0$ go to infinity.  Corollary \ref{corollary1} demonstrates that not only can the proposed SGF transmission scheme  avoid these error floors, but also can  it   ensure that the achievable diversity gain is proportional to the number of participating grant-free users, i.e., the more grant-free users there are, the better the   outage performance.

{\it Remark 9:} The main reason why the proposed SGF scheme   avoids an error floor can be  explained as follows.  By using  \eqref{out overall}, an upper bound on the outage probability achieved by the proposed SGF scheme can be obtained  as follows:
\begin{align}  
{\rm P}_{out} = & \sum^{M}_{m=1}{\rm P}\left(E_m,     \max \left\{R_{m,I},R_{M,II} \right\}<R_s\right)+ {\rm P}\left(E_0,  R_{M,II} <R_s\right)
\\\nonumber
 \leq  & \sum^{M}_{m=1}{\rm P}\left( R_{m,I} <R_s\right)+ {\rm P}\left(E_0,  R_{M,II} <R_s\right)
 \\\nonumber
 =  & \sum^{M}_{m=1}{\rm P}\left( \log\left(1+P_s|h_m|^2\right)<R_s\right)+ {\rm P}\left(E_0,  \log\left(1+\frac{P_s|h_M|^2}{P_0|g|^2+1}\right) <R_s\right),
\end{align}
where the last step follows from the definitions of $R_{m,I}$ and $R_{M,II}$. 

By using the fact that the users are ordered as in \eqref{channel order}, ${\rm P}_{out}$ can be further upper bounded as follows: 
 \begin{align}  \label{lax}
{\rm P}_{out} \leq  & M{\rm P}\left( \log\left(1+P_s|h_1|^2\right)<R_s\right)+ \underset{Q_u}{\underbrace{{\rm P}\left(E_0,  \log\left(1+\frac{P_s|h_M|^2}{P_0|g|^2+1}\right) <R_s\right)}}.
\end{align}

Recall that the outage probability for SGF Scheme II is ${\rm P}_{out}^{II}\triangleq {\rm P}\left(\log\left(1+\frac{P_s|h_M|^2}{P_0|g|^2+1}\right)<R_s\right)$, where an error floor exists since its signal-to-interference-plus-noise ratio (SINR) becomes a constant when $P_s$ and $P_0$ go to infinity. The probability   $Q_u$ is quite similar to ${\rm P}_{out}^{II}$, but $Q_u$  does not exhibit  an  error floor, as explained in the following. By using the definition of $E_0$,    $Q_u$ can be rewritten 
as follows:
\begin{align}
Q_u=&{\rm P}\left(|h_{1}|^2>\frac{\tau(|g|^2)}{P_s} ,    |h_M|^2  <\alpha_s(P_0|g|^2+1)\right)  .
\end{align}
Since $|h_{1}|^2\leq |h_{M}|^2$, the lower bound on $|h_{1}|^2$,  $\frac{\tau(|g|^2)}{P_s}$, needs to be smaller than the upper bound on $|h_M|^2$, $\alpha_s(P_0|g|^2+1)$, which introduces  an additional constraint $|g|^2< \alpha_2$, if $\epsilon_s\epsilon_0< 1$, as shown in \eqref{range} - \eqref{range3}.  This additional  constraint $|g|^2< \alpha_2$   effectively removes the error floor since
\begin{align}
Q_u =&{\rm P}\left(|h_{1}|^2>\frac{\tau(|g|^2)}{P_s} ,  \log\left(1+\frac{P_s|h_M|^2}{P_0|g|^2+1}\right)<R_s,|g|^2<\alpha_2 \right) \\\nonumber \leq& {\rm P}\left(|g|^2<\alpha_2 \right) = 1-e^{-\alpha_2}\rightarrow 0,
\end{align}
for $P_0\rightarrow \infty$. On the other hand, it is straightforward to show that the first term in \eqref{lax}, ${\rm P}\left( \log\left(1+P_s|h_1|^2\right)<R_s\right)$, also goes to zero at high SNR. Therefore, ${\rm P}_{out}$ does not have an error floor.

\section{Simulation Results} \label{section IV}
In this section, the performance of the proposed  SGF transmission scheme is studied via  computer simulations, where the accuracy of the developed analytical results is also evaluated. To facilitate performance evaluation, the two existing SGF schemes proposed in \cite{8662677} are used as   benchmark schemes. We note that the proposed SGF scheme allows all the users to participate in contention. Therefore, for a fair comparison, we choose $\tau_{I}=\infty$ and $\tau_{II}=0$ for the two benchmarking schemes, which allow all   grant-free users to participate in contention and hence yield the best performance for the two schemes.

In Fig. \ref{fig 1}, the outage performance achieved by the proposed SGF transmission scheme is compared to those of the two existing schemes for different choices of $P_s$ and $P_0$. In particular, in Fig. \ref{fig1a}, we assume $P_s=\frac{P_0}{10}$, which is equivalent to the case where the grant-free users have weaker channel conditions than the grant-based user, if all the users use the same transmit power. Recall that SGF Scheme I  first decodes the grant-based user's signal by treating the grant-free user's signal as interference. Therefore, the situation with $P_s=\frac{P_0}{10}$ is ideal for the application of SGF Scheme I, and Fig. \ref{fig1a} confirms this conclusion since SGF Scheme I outperforms SGF Scheme II. We note that for SGF Scheme I, the outage probability for $M=1$ can be better than that for $M=5$, since a larger $M$ can reduce  ${\rm P}\left(\log\left(1+\frac{P_0|g|^2}{P_s|h_1|^2+1}\right)<R_0\right)$  but may increase ${\rm P}\left(\log\left(1+P_s|h_1|^2\right)<R_s\right)$.  In Fig. \ref{fig1b}, we focus on the situation, where  $P_s\rightarrow \infty$ and $P_0$ is kept  constant. This is equivalent to the case where the grant-free users have stronger channel conditions than the grant-based user, if all the users use the same transmit power. Therefore,  this situation is ideal for the application of SGF Scheme II, and Fig. \ref{fig1b} shows that  SGF Scheme II indeed outperforms SGF Scheme I. For both considered scenarios,   the proposed  SGF scheme   outperforms the two existing schemes, and can also effectively avoid   error floors in both considered  scenarios, as shown in the two figures. 
 
\begin{figure}[t] \vspace{-2em}
\begin{center}\subfigure[$P_s=\frac{P_0}{10}$ ]{\label{fig1a}\includegraphics[width=0.48\textwidth]{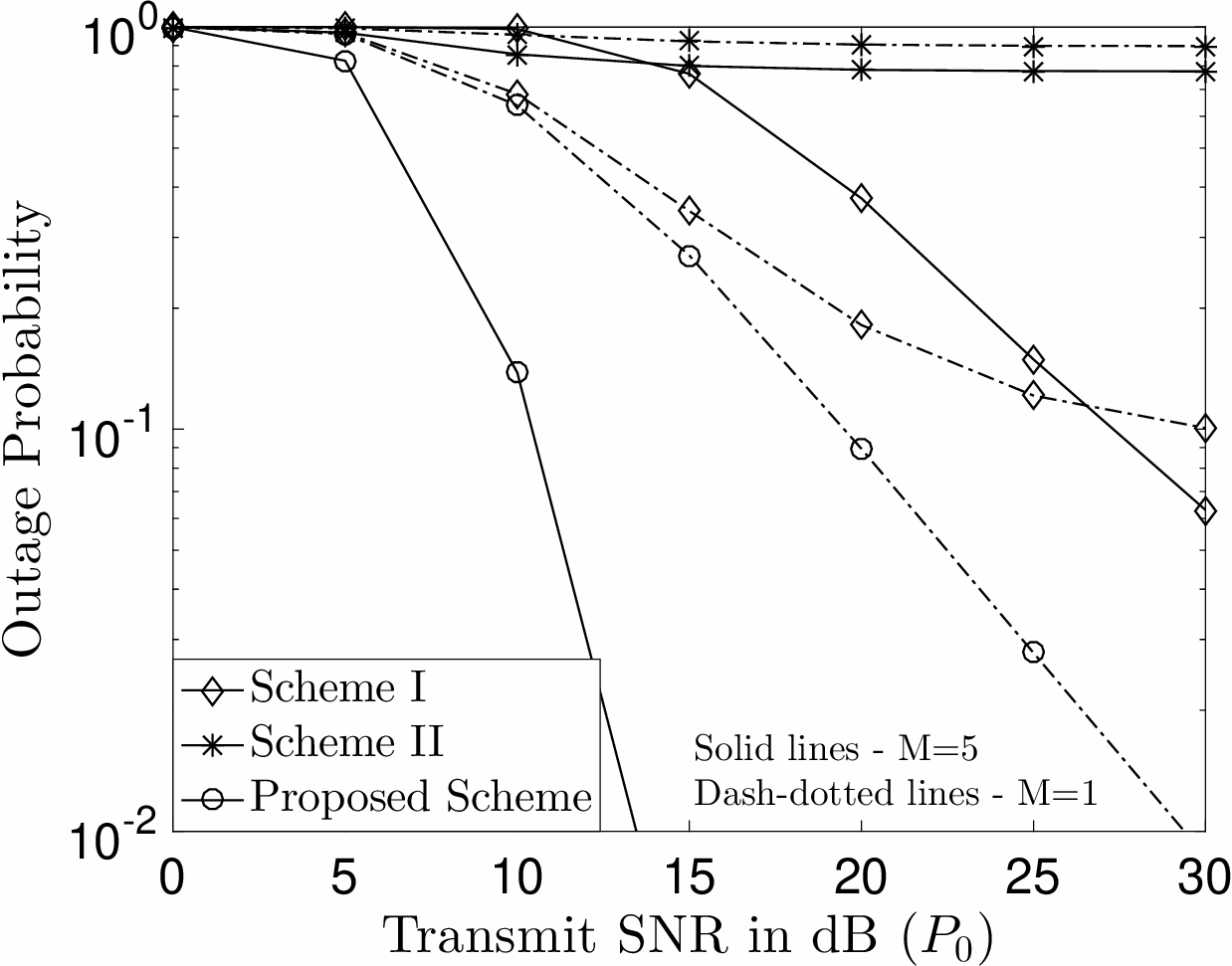}}
\subfigure[$P_0=10$ dB]{\label{fig1b}\includegraphics[width=0.48\textwidth]{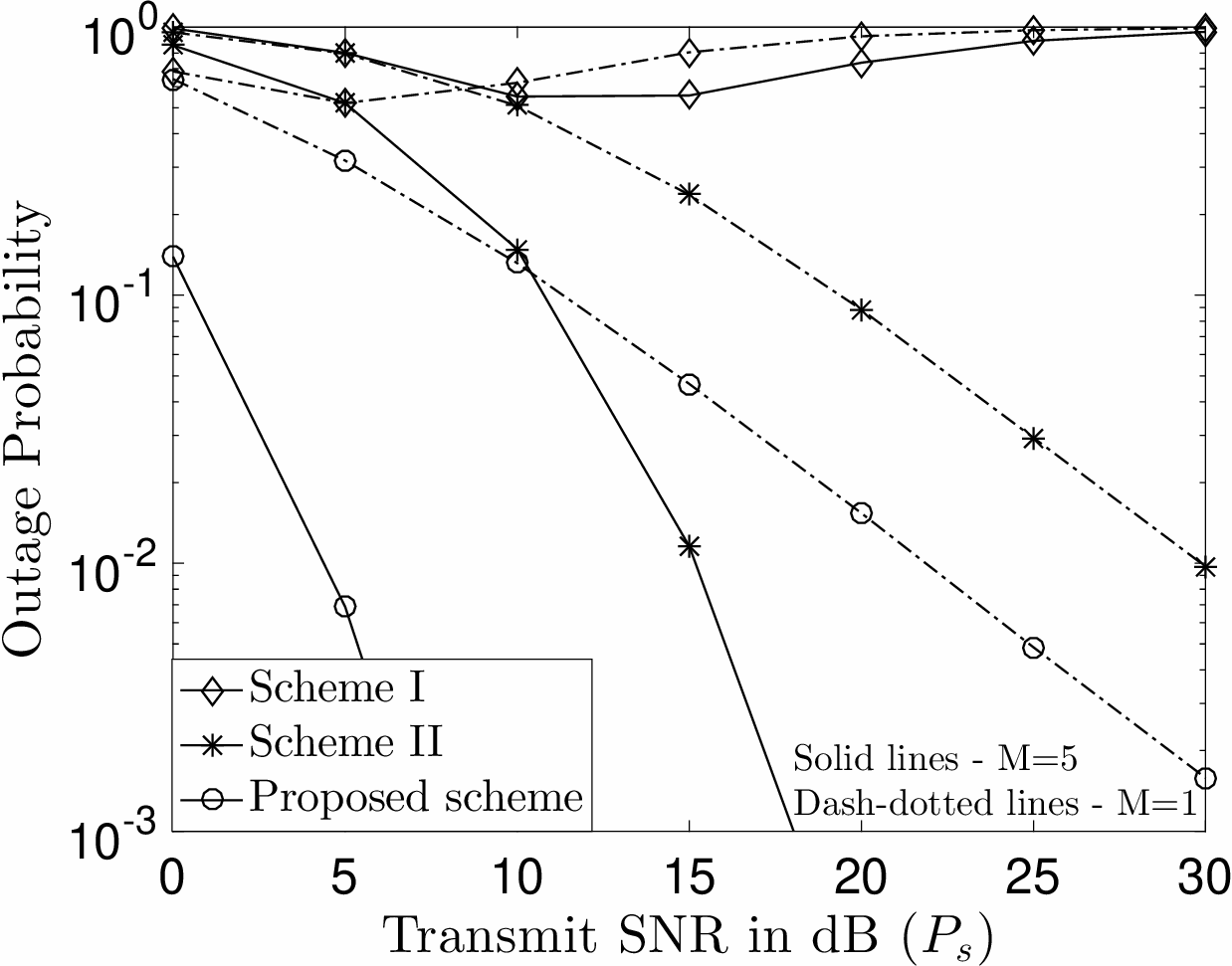}} \vspace{-1em}
\end{center} 
 \caption{Comparison of the three SGF transmission schemes.   $R_0=1$ bit per channel use (BPCU), and $R_s=0.9$ BPCU.      }\label{fig 1}\vspace{-1.5em}
\end{figure}

In Fig. \ref{fig 2}, we  examine the accuracy of the developed analytical results for the outage probability. In Fig. \ref{fig2a}, the exact expressions for the outage probabilities shown in Theorem \ref{theorem1} and Corollary \ref{corollary0} are used, and the figure shows that the curves for the analytical results perfectly match the curves obtained by simulations, which demonstrates the accuracy of the result provided  in Theorem \ref{theorem1}.  In Fig. \ref{fig2b}, the accuracy of the approximations developed in Theorem \ref{theorem2} and Corollary \ref{corollary1} is studied. As can be observed from the figure, both approximations are accurate at high SNR. We note that the approximation   in Corollary \ref{corollary1} becomes less accurate as  $M$ increases. This is due to the fact that the approximation in Corollary \ref{corollary1} disregards   the terms, $Q_m$, $0\leq m \leq M-1$, and $Q_{M+1}$, and considers  $Q_M$ only. When $M$ is small, such an approximation is accurate. But the gap between the approximation and the actual  probability becomes noticeable when $M$ becomes large. The fact that the curves for the approximation in Corollary \ref{corollary1}  are   below the other curves is also due to the same reason. 

 In Fig. \ref{fig3}, the impact of different choices for the target data rate and the transmit power   on the outage performance is studied. The figure shows that reducing the grant-free user's target rate   can affect the outage probability   more significantly than reducing the grant-based user's target rate. In addition, the figure shows that, for a fixed $P_0$, increasing $P_s$ can improve  the  grant-free user's outage performance, i.e., a grant-free user can improve its performance by increasing its own transmit power. This is not true for SGF Scheme I since increasing  $P_s$ deteriorates the probability ${\rm P}\left(\log\left(1+\frac{P_0|g|^2}{P_s|h_1|^2+1}\right)<R_0\right)$. 

\begin{figure}[t] \vspace{-2em}
\begin{center}\subfigure[Exact analytical results  ]{\label{fig2a}\includegraphics[width=0.48\textwidth]{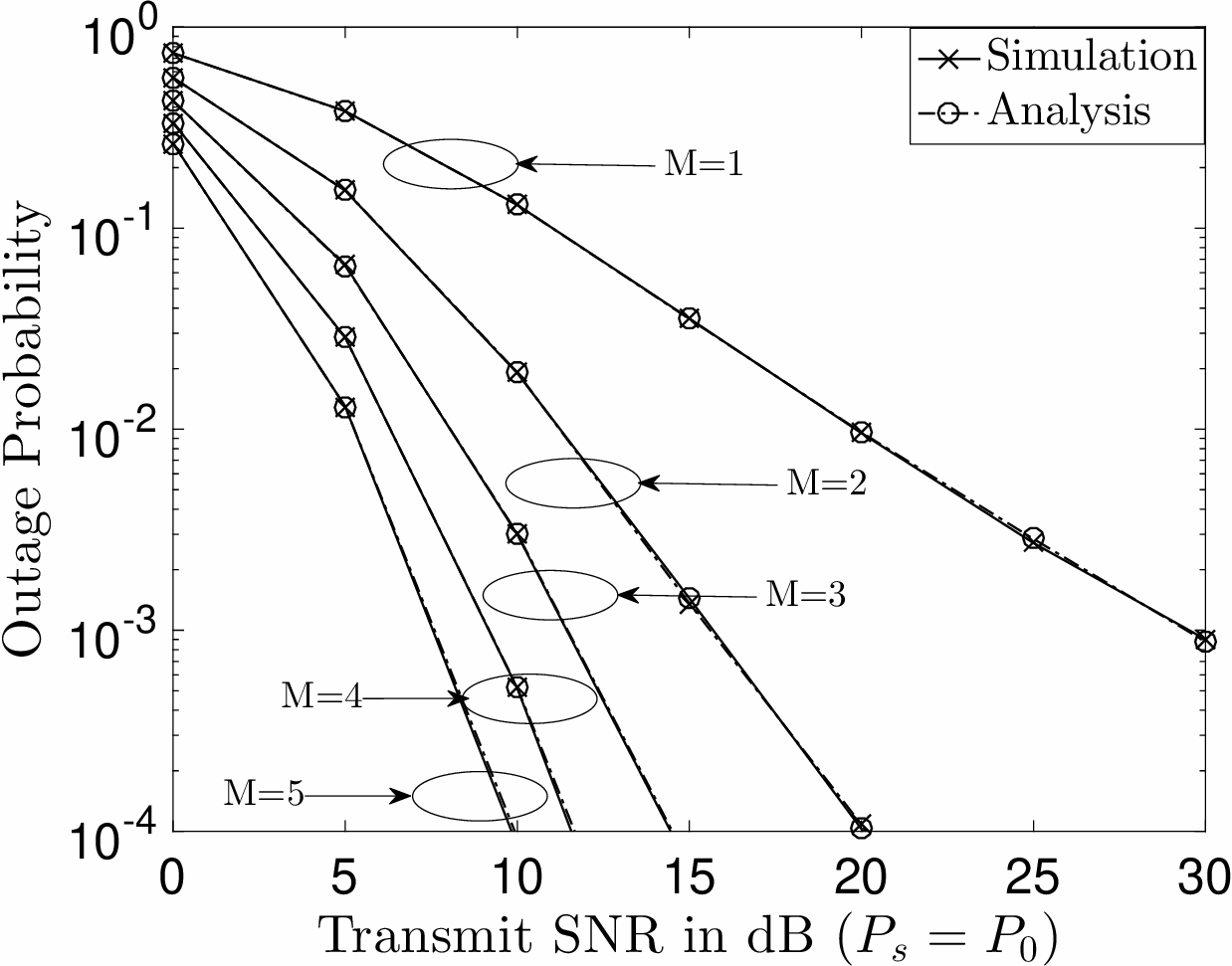}}
\subfigure[Approximation results]{\label{fig2b}\includegraphics[width=0.48\textwidth]{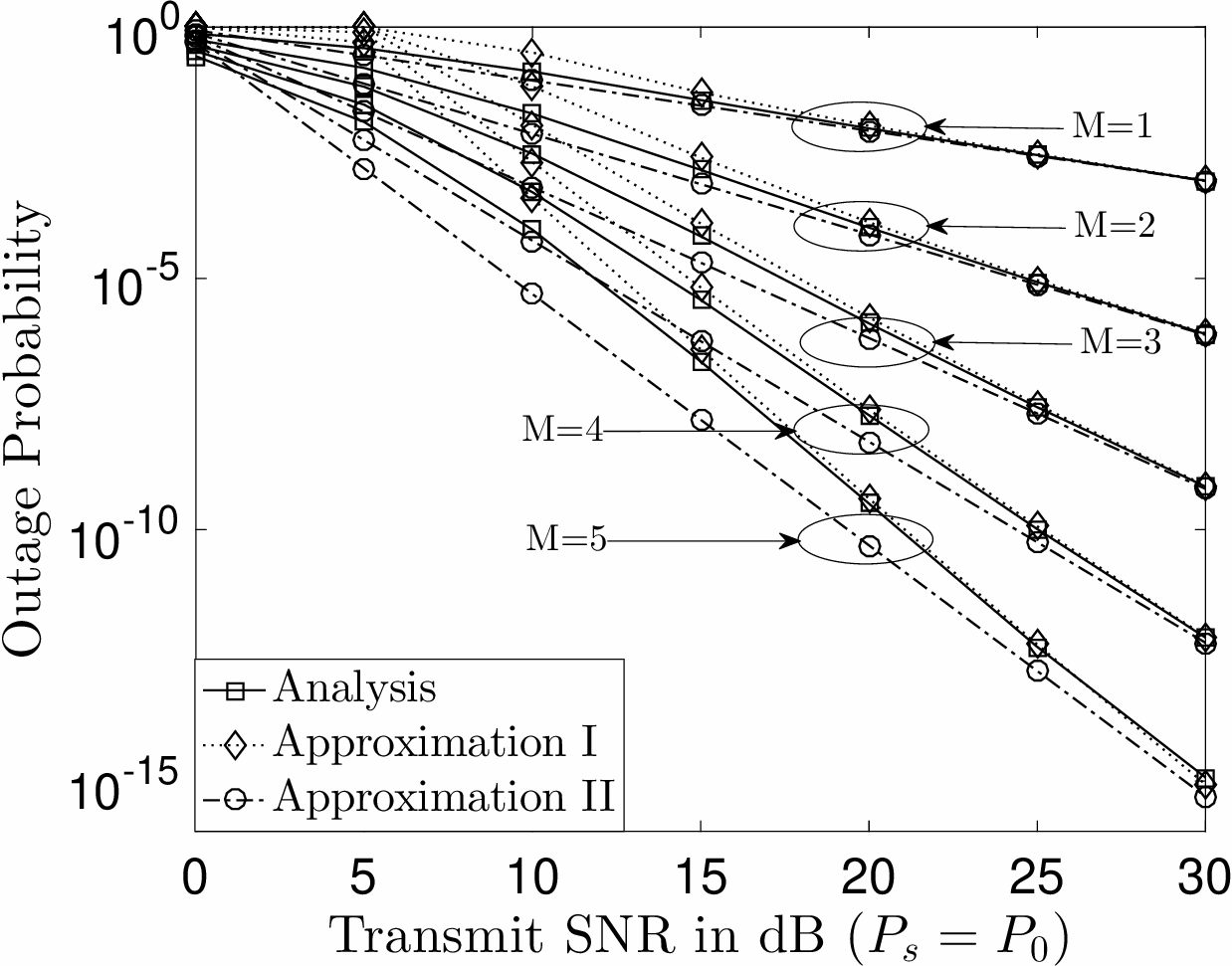}} \vspace{-1em}
\end{center} 
 \caption{Accuracy of the developed analytical results. $P_0=P_s$, $R_0=1$ bit per channel use (BPCU), and $R_s=0.9$ BPCU.  The curves for Analysis are based on Theorem \ref{theorem1} and Corollary \ref{corollary0}, the curves   for Approximation I are based on Theorem \ref{theorem2}, and the curves  for Approximation II are based on Corollary \ref{corollary1}.   }\label{fig 2}\vspace{-0em}
\end{figure}

\begin{figure}[!t]\centering \vspace{-1em}
    \epsfig{file=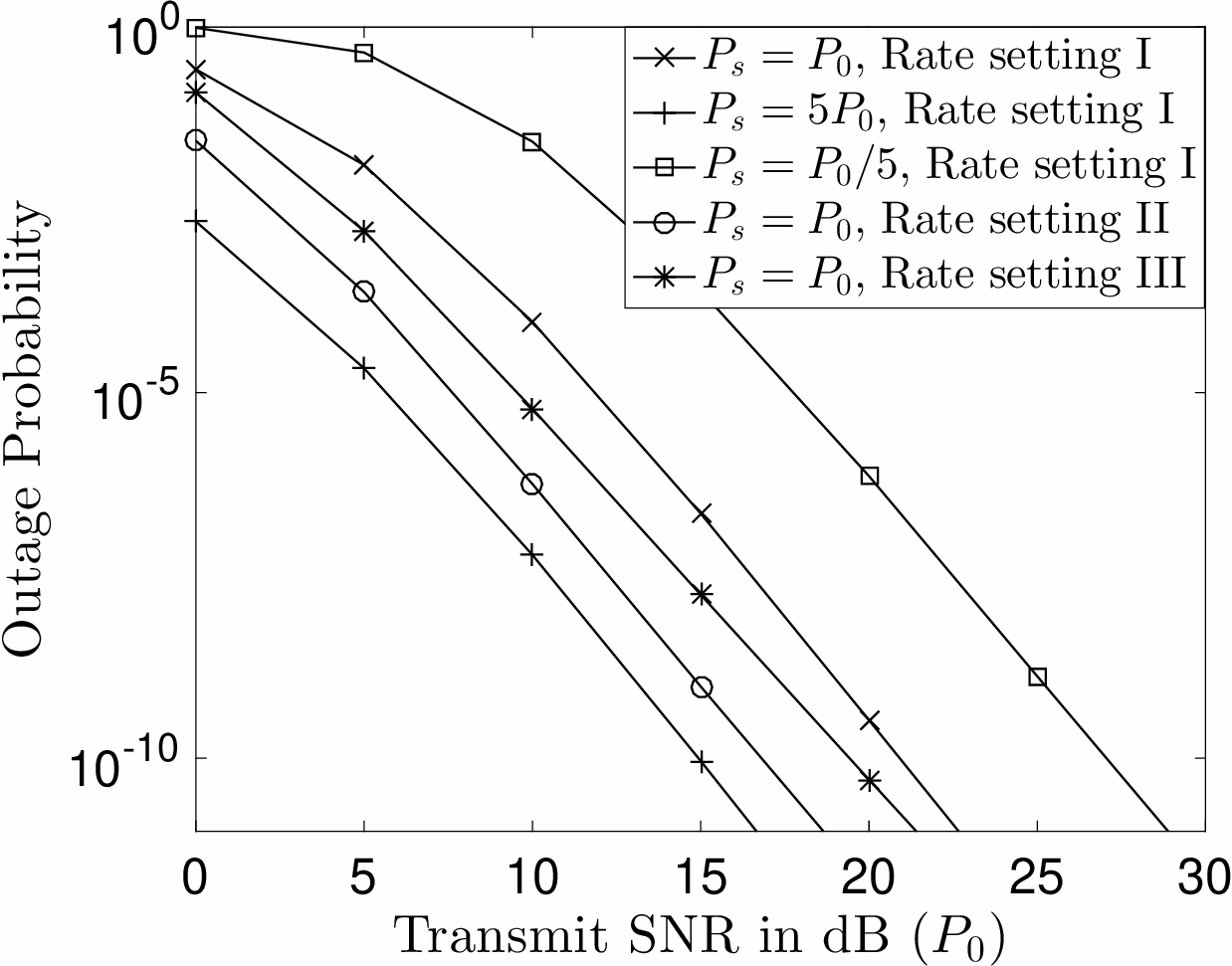, width=0.48\textwidth, clip=}\vspace{-0.5em}
\caption{ Performance of the proposed SGF transmission scheme for different choices for the transmit power and the target rates. For Rate Setting I, $R_0=1$ BPCU and $R_s=0.9$ BPCU. For Rate Setting II, $R_0=1$ BPCU and $R_s=0.5$ BPCU. For Rate Setting III,  $R_0=0.5$ BPCU and $R_s=0.9$ BPCU.  $M=5$.    \vspace{-1em} }\label{fig3}\vspace{-0.5em}
\end{figure}

\begin{figure}[!t]\centering \vspace{-0.1em}
    \epsfig{file=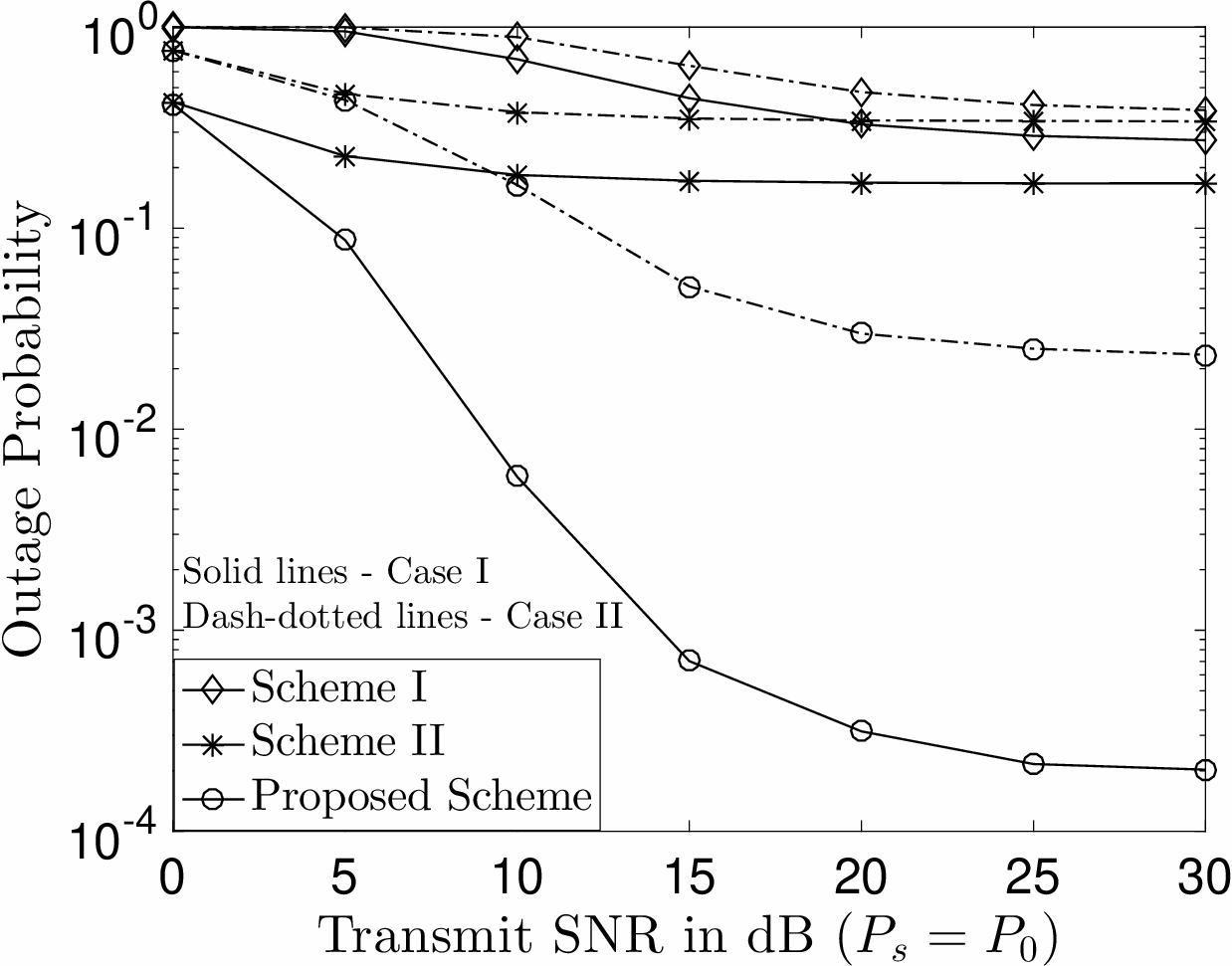, width=0.48\textwidth, clip=}\vspace{-0.5em}
\caption{ Performance of the proposed SGF transmission scheme when $\epsilon_s\epsilon_0\geq 1$.   For Case I, $R_0=1.5$ BPCU, $R_s=1$ BPCU. For Case II, $R_0=2$ BPCU, $R_s=1.5$ BPCU. $M=5$.      \vspace{-1em} }\label{fig4}\vspace{-0.5em}
\end{figure}

In Fig. \ref{fig4}, the performance of the proposed SGF transmission scheme is evaluated under the condition that $\epsilon_s\epsilon_0\geq 1$. As discussed in Remarks 3 and 7, the condition $\epsilon_s\epsilon_0< 1$ is important to avoid   error floors. If this condition does not hold, error floors appear, as shown in Fig. \ref{fig4}. However, we note that the outage performance achieved by the proposed SGF transmission scheme is still significantly better than those of the two existing schemes. For example, for the case with $R_0=1.5$ bits per channel use (BPCU), and $R_s=1$ BPCU, the   proposed  scheme can achieve  an outage probability of $1\times 10^{-4}$, whereas the outage probabilities achieved by the two existing schemes exceed  $1\times 10^{-1}$.

\begin{figure}[!t]\centering \vspace{-0em}
    \epsfig{file=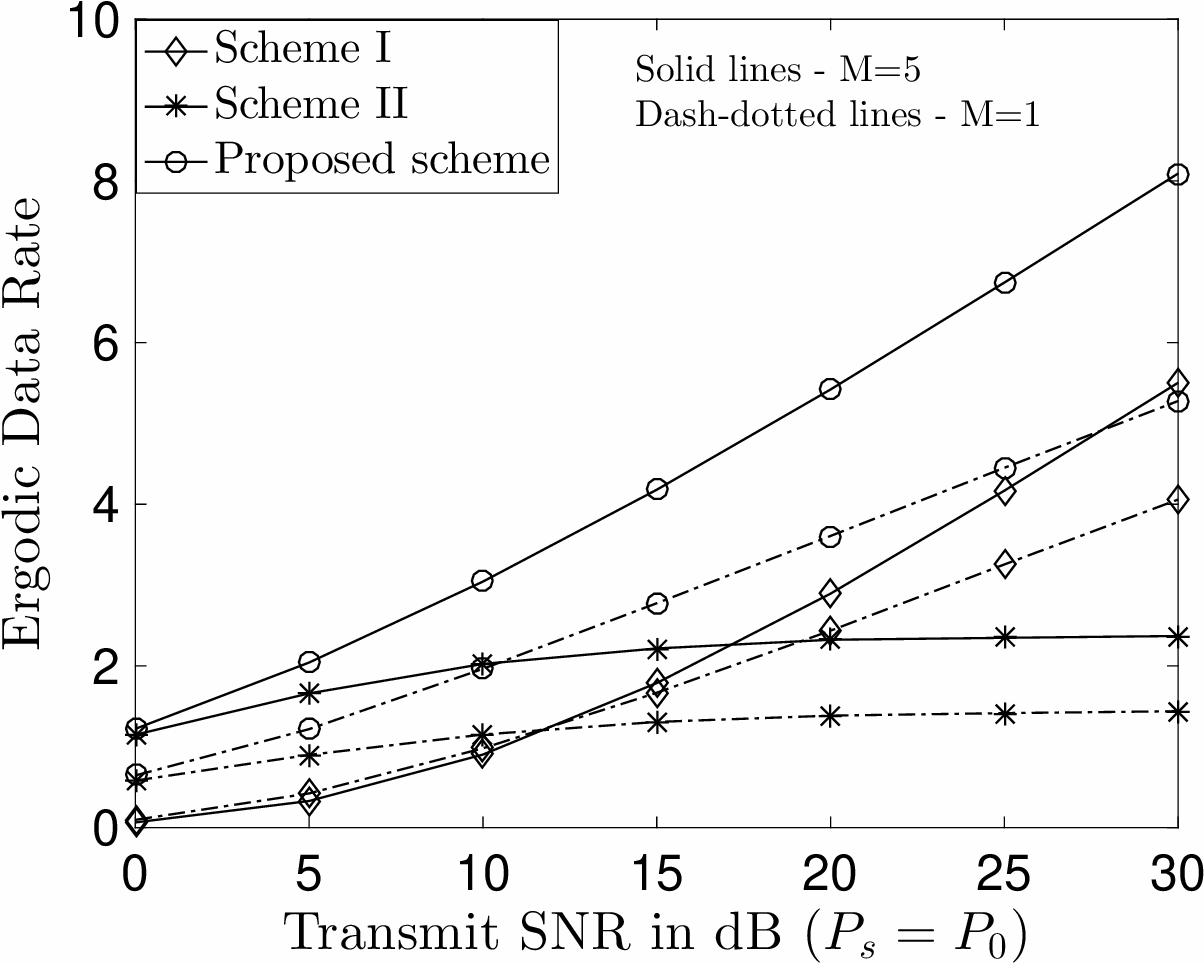, width=0.48\textwidth, clip=}\vspace{-0.5em}
\caption{ Ergodic data rate comparison of the three considered SGF transmission schemes. $R_0=1$ BPCU.     \vspace{-1em} }\label{fig5}\vspace{-0.5em}
\end{figure}

In Fig. \ref{fig5}, the ergodic data rate is used to evaluate the performance of the considered SGF schemes. As can be observed from the figure, the proposed SGF scheme outperforms both existing schemes, particularly at high SNR, which is consistent with   the figures showing  the outage probability. In addition, Fig. \ref{fig5} shows that the slope of the curves for the proposed SGF scheme is larger than those of the two existing schemes, which demonstrates  that the proposed scheme can effectively exploit    multi-user diversity.  An interesting observation is that, for high  transmit powers, the curves for SGF Scheme II become flat, whereas the curves for the other two schemes do not. This is due to the fact that the    data rate achieved by Scheme II is $  \log\left(1+\frac{P_s|h_M|^2}{P_0|g|^2+1}\right) $, which becomes a constant when both $P_s$ and $P_0$ approach infinity. On the other hand, once the grant-based user's target data rate can be realized, the achievable data rates for SGF Scheme I and the proposed SGF scheme are of the form   $\log(1+P_s|h_m|^2)$, which means that their ergodic data rates are not bounded when $P_s$ goes to infinity, as confirmed by the figure. The performance gain of the proposed scheme over Scheme I is due to the fact that, when the grant-based user's signal cannot be decoded correctly in the first stage of SIC, the data rate of Scheme I becomes zero, but the proposed scheme can still offer a non-zero data rate by changing  the SIC order.

 Fig. \ref{fig 6} shows  the grant-free users' admission probabilities, i.e.,  which  grant-free user is admitted to the resource block by the proposed SGF scheme, for different choices of $R_0$ and the users' transmit powers. We first note that the admission probability for grant-free user $m$   is given by
\begin{align} 
{\rm P}_m = &  {\rm P}\left(E_m,     R_{m,I}>R_{M,II}  \right)  ,
\end{align}
for $1\leq m \leq M-1$, and 
\begin{align}\nonumber
{\rm P}_M = &\sum^{M-1}_{n=1}  {\rm P}\left(E_m,     R_{m,I}<R_{M,II}  \right) + {\rm P}\left(E_M   \right)+ {\rm P}\left(E_0 \right).
\end{align}
 Fig. \ref{fig 6} shows that at low SNR, grant-free user $M$, the user with the strongest channel gain, is preferred over the other users, as explained as follows.  At low SNR, the threshold $\tau(|g|^2) \triangleq \max\left\{0, \frac{P_0|g|^2}{2^{R_0}-1}-1\right\}
$ is very likely to be zero, which means that Group 2 is empty, i.e., $E_0$ happens. As a result, grant-free user $M$ is granted access. In addition,  Fig. \ref{fig 6} shows that at high SNR, the users' admission probabilities become constant, and increasing $R_0$ increases the admission probabilities of the grant-free users whose channel gains are weak, which can be explained as follows.  By assuming $P_s=P_0\rightarrow \infty$, ${\rm P}_m$, $1\leq m \leq M-1$, can be approximated as follows:\footnote{Obtaining an exact expression for ${\rm P}_m$ is not a trivial task since ${\rm P}_m$ is a function of four random variables, $|g|^2$, $|h_m|^2$, $|h_{m+1}|^2$ and $|h_M|^2$. The fact that $|h_m|^2$, $|h_{m+1}|^2$ and $|h_M|^2$ are not independent makes it   more difficult  to analyze ${\rm P}_m$, which is an important direction for future research. }
\begin{align} \label{pmxx}
{\rm P}_m = &  {\rm P}\left(E_m,     R_{m,I}>R_{M,II}  \right) \rightarrow  {\rm P}\left(E_m \right)\\\nonumber =& {\rm P}\left(
 |h_{m}|^2<\frac{\tau(|g|^2)}{P_s},    |h_{m+1}|^2> \frac{\tau(|g|^2)}{P_s}  \right)
  \rightarrow  {\rm P}\left(
 |h_{m}|^2< \frac{|g|^2}{2^{R_0}-1} ,    |h_{m+1}|^2>  \frac{|g|^2}{2^{R_0}-1}   \right),
\end{align}
which is indeed a constant and not a function of $P_s$ or $P_0$. Since  ${\rm P}_m$, $1\leq m \leq M-1$, are constant at high SNR and $\sum^{M}_{m=1}{\rm P}_m =1$, ${\rm P}_M $ must also be constant at high SNR, as confirmed by the figure. By increasing $R_0$, $ \frac{|g|^2}{2^{R_0}-1} $ is reduced, and \eqref{pmxx} indicates that  ${\rm P}_m$ is increased for  small $m$, i.e., the weak users' admission probabilities are increased, as shown in Fig. \ref{fig 6}.

\begin{figure}[t] \vspace{-0.5em}
\begin{center}\subfigure[$R_0=0.5$ BPCU ]{\label{fig6a}\includegraphics[width=0.48\textwidth]{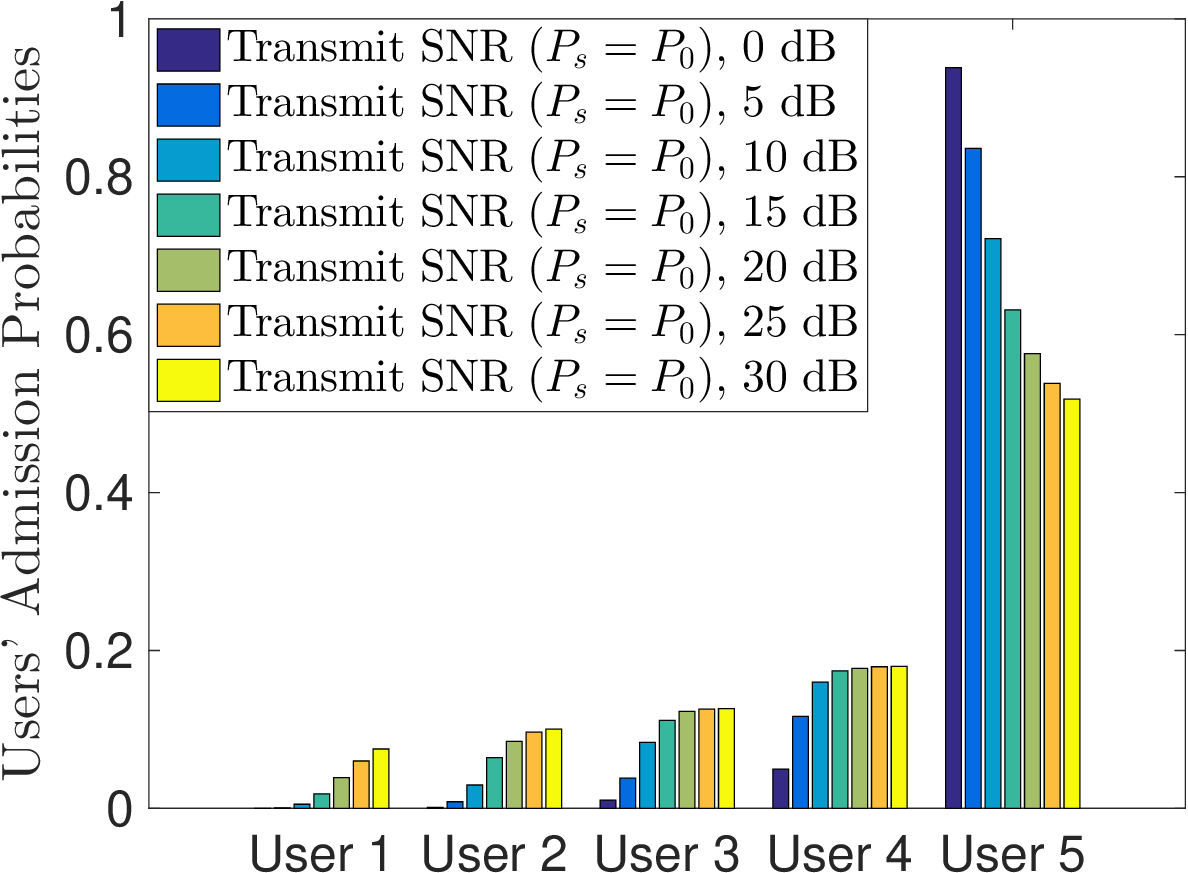}}
\subfigure[$R_0=1.5$ BPCU]{\label{fig6b}\includegraphics[width=0.48\textwidth]{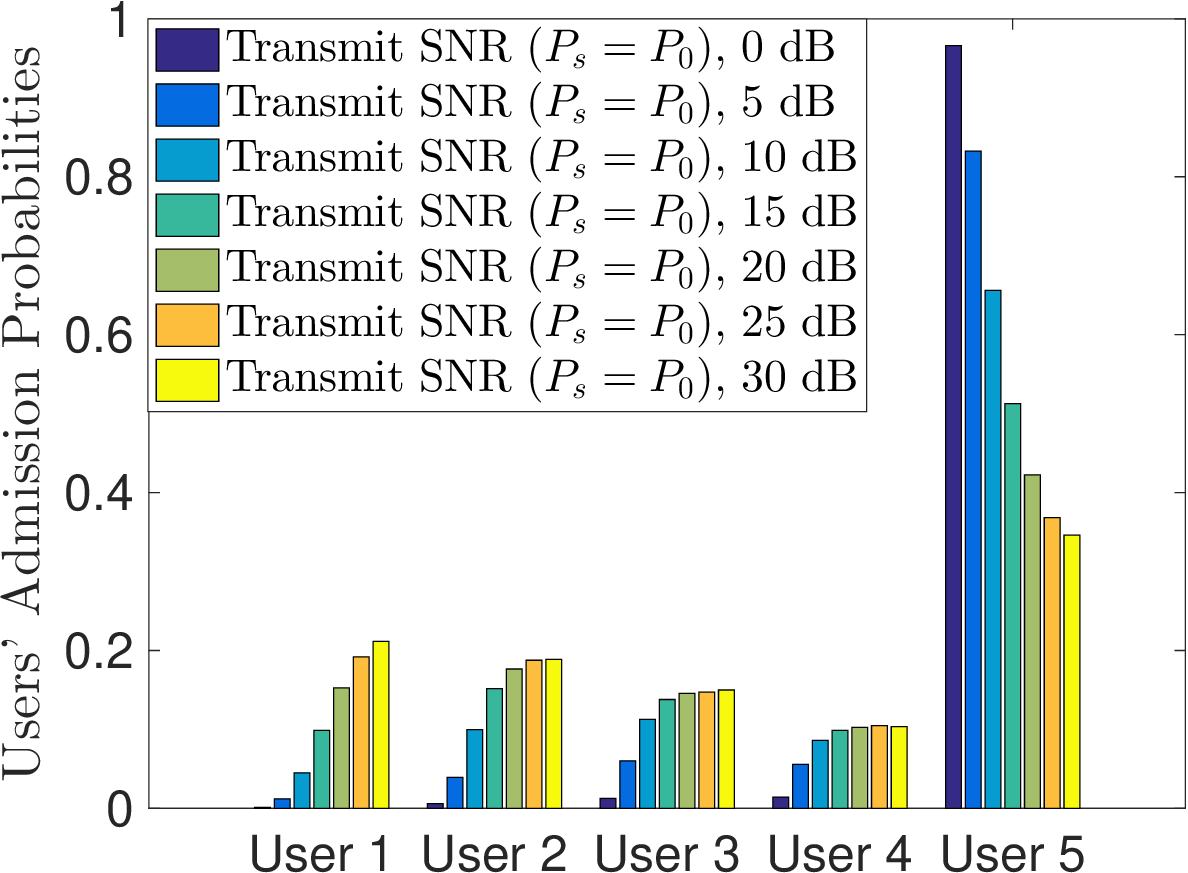}} \vspace{-1em}
\end{center} 
 \caption{Users' admission probabilities for the proposed SGF scheme. $M=5$.  }\label{fig 6}\vspace{-0em}
\end{figure}

\section{Conclusions} \label{section V}
In this paper, we have proposed a new NOMA assisted SGF transmission scheme. Compared to the two existing SGF schemes, this new scheme can ensure that admitting a grant-free user is completely transparent to the grant-based user, i.e., the grant-based user communicates with its base station as if it solely occupied the channel. In addition, the proposed SGF scheme significantly improves  the  reliability  of the grant-free users' transmissions  compared to the existing SGF schemes. To facilitate the performance evaluation of the proposed SGF scheme, an exact expression for the outage probability   was derived, where an asymptotic analysis  was also   carried out to show that the full  multi-user diversity gain of $M$ is achievable. Computer simulation results were provided to demonstrate the performance of the proposed SGF scheme and  to verify the accuracy of the developed analytical results.  

In this paper, Rayleigh fading is assumed for the users' channel gains. An important direction for future research is to carry out a stochastic geometry analysis by taking the users' path  losses  into consideration.  In addition, we assumed  that the admitted grant-free user uses only one time slot for transmission. In practice, the grant-free user may perform   retransmission and occupy the channel for a few consecutive time slots. An interesting  direction for  future research is to develop  hybrid automatic repeat request (HARQ) schemes for SGF transmission. 
%

\appendices

\section{Proof for Theorem \ref{theorem1}}
The evaluation of probability $Q_m$     in \eqref{overall} depends on the value of $m$,  as shown in the following subsections.  

\subsection{Evaluation of  $Q_m$, $1\leq m\leq M-2$}
In this case,   probability   $Q_m$ involves  three order statistics, $h_m$, $h_{m+1}$,  and $h_M$, and  can be expressed as follows:
\begin{align} \nonumber 
  {Q}_{m} =&{\rm P}\left(E_m,|g|^2>\frac{\epsilon_0}{P_0},  \log\left(1+P_s|h_m|^2\right)<R_s,  \log\left(1+\frac{P_s|h_M|^2}{P_0|g|^2+1}\right) <R_s \right)
 \\  =&\underset{|g|^2>\frac{\epsilon_0}{P_0}} {\mathcal{E}}\left\{{\rm P}\left( |h_m|^2 <\xi, |h_{m+1}|^2 >\frac{P_0\epsilon_0^{-1}|g|^2-1}{P_s} ,  |h_M|^2<\frac{\epsilon_s(1+P_0|g|^2) }{P_s} \right)\right\}, \label{qmz1}
\end{align}
where  $\mathcal{E}\{\cdot\}$ denotes the expectation operation, and $\xi=\min\left\{\frac{\epsilon_s}{P_s}, \frac{ P_0\epsilon_0^{-1}|g|^2-1}{P_s}\right\}$.

For the case $1\leq m\leq M-2$, $h_{m+1}$  and $h_M$ are different. As a result, there is   a hidden constraint  in \eqref{qmz1} that the lower bound on $h_{m+1}$ should be smaller than the upper bound on $h_M$, i.e., $\frac{\epsilon_s(1+P_0|g|^2) }{P_s}>\frac{P_0\epsilon_0^{-1}|g|^2-1}{P_s}$.  We first note that whether $\frac{\epsilon_s(1+P_0|g|^2) }{P_s}>\frac{P_0\epsilon_0^{-1}|g|^2-1}{P_s}$ holds depends on the value of $g$ as shown in the following: 
\begin{align}\label{range}
 &\epsilon_s (1+P_0|g|^2)- \left(P_0\epsilon_0^{-1}|g|^2-1\right)
 \\\nonumber =&(\epsilon_s-\epsilon_0^{-1})  P_0|g|^2 +\epsilon_s+1
 \left\{ \begin{array}{ll}  <0, & {\rm if} \text{ } |g|^2 >\frac{\epsilon_0(\epsilon_s+1)}{(1-\epsilon_0\epsilon_s)  P_0}\\ >0, & \text{otherwise}
 \end{array}\right.,
\end{align}
where the assumption   that $\epsilon_s\epsilon_0<1$ was used.
Furthermore, we note that the following inequality always holds
\begin{align}\label{range2}
\frac{\epsilon_0(\epsilon_s+1)}{(1-\epsilon_0\epsilon_s)  P_0}>\frac{\epsilon_0}{P_0},
\end{align}
 since 
$(\epsilon_s+1 )- (1-\epsilon_s\epsilon_0 ) = \epsilon_s+\epsilon_s\epsilon_0\geq 0$.

Therefore, denoting the probability inside   the expectation in \eqref{qmz1} by $S_m$,     $Q_m$ can be expressed as follows:  
\begin{align} \label{range3}
  {Q}_{m} =&\underset{ \alpha_2>|g|^2>\alpha_0} {\mathcal{E}}\left\{ S_m\right\}+\underset{|g|^2>\alpha_1} {\mathcal{E}}\left\{ S_m\right\}\\\nonumber
 =&\underset{\alpha_2>|g|^2>\alpha_0} {\mathcal{E}}\left\{ S_m\right\},
\end{align}
where the last step follows by using \eqref{range}.

For the case $1\leq m \leq M-2$, $S_m$ is a function of three order statistics,  $|h_m|^2$, $|h_{m+1}|^2$, and $|h_{M}|^2$, whose joint probability density function (pdf) is given by \cite{David03}
\begin{align}\label{3pdf}
&f_{|h_{m}|^2,|h_{m+1}|^2, |h_{M}|^2}(x,y,z) \\\nonumber=& \eta_m e^{-x}\left(1-e^{-x}\right)^{m-1}   e^{-y}\left(e^{-y}-e^{-z}\right)^{M-m-2} e^{-z}\\\nonumber
 =& \eta_m \sum^{M-m-2}_{i=0}{M-m-2 \choose i}(-1)^ie^{-x}\left(1-e^{-x}\right)^{m-1}  e^{-y} e^{-(M-m-2-i)y}e^{-iz}\  e^{-z},
\end{align}
where $x\leq y\leq z$ and $\eta_m = \frac{M!}{(m-1)!(M-m-2)!}$.

For a fixed $g$ and by using the joint pdf shown in \eqref{3pdf},   $S_m$  can be expressed as follows:
\begin{align}\nonumber
S_m =&  \eta_m \sum^{M-m-2}_{i=0}{M-m-2 \choose i}(-1)^i\int_0^{\xi}e^{-x}\left(1-e^{-x}\right)^{m-1}\\\nonumber &\times  \int_{\frac{\alpha_0^{-1}|g|^2-1}{P_s} }^{\alpha_s(1+P_0|g|^2)}  e^{-(M-m-1-i)y}    \int^{\alpha_s(1+P_0|g|^2) }_{y}
    e^{-(i+1)z} dz dydx. 
\end{align}

With some algebraic manipulations, $S_m$ can be calculated as follows:
\begin{align}  
S_m \nonumber
  =
   \eta_m \sum^{M-m-2}_{i=0}{M-m-2 \choose i}  \sum^{m}_{p=0}{m \choose p}\frac{(-1)^{i+p}e^{-p\xi}}{m(i+1)}  \\  \times   \nonumber
  \left(\frac{\mu_3e^{-\mu_1 |g|^2}-\mu_5e^{-\mu_6 |g|^2}
  }{M-m}  - \frac{\mu_4  e^{-\mu_2|g|^2}-\mu_5  e^{-\mu_{6}|g|^2}}{M-m-1-i}\right),
\end{align} 
where $\mu_1 = \frac{(M-m)\alpha_0^{-1}}{P_s}$, $\mu_2 = \left((i+1)\alpha_s+(M-m-1-i)\frac{\epsilon_0^{-1}}{P_s} \right)P_0$, $\mu_3=e^{\frac{M-m}{P_s} }$, $\mu_4 = e^{-(i+1)\alpha_s+(M-m-1-i)\frac{ 1}{P_s} }
$, $\mu_5=e^{-(M-m)\alpha_s  }$, and $\mu_6=(M-m)\alpha_sP_0$.  
 
 Recall that $Q_m$ can be obtained by finding the expectation of $S_m$ for $\alpha_2>|g|^2>\alpha_0$, i.e., $ {Q}_{m} =\underset{\alpha_2>|g|^2>\alpha_0} {\mathcal{E}}\left\{ S_m\right\}$.  
We note that   $S_m$ is a function   $|g|^2$ via $\xi$. The complication is that   $\xi$ can have two possible forms depending on the value  of $|g|^2$ as shown in the following: 
 \begin{align}
 \xi = &\left\{ \begin{array}{ll} \alpha_s, &\text{if } \epsilon_s<\alpha_0^{-1}|g|^2-1\\  \frac{\alpha_0^{-1}|g|^2-1}{P_s}, &\text{otherwise} \end{array}\right.
 \\\nonumber 
 = &\left\{ \begin{array}{ll} \alpha_s, &\text{if } |g|^2>  \alpha_1 \\  \frac{\alpha_0^{-1}|g|^2-1}{P_s}, &\text{otherwise} \end{array}\right..
 \end{align}

 It is important to note that $\alpha_0\leq \alpha_1\leq \alpha_2$     always holds since
\begin{align}
\frac{\epsilon_0}{P_0} \leq \frac{\epsilon_0(1+\epsilon_s)}{P_0} \leq \frac{\epsilon_0(1+\epsilon_s)}{P_0(1-\epsilon_0\epsilon_s)}. 
\end{align}

Therefore, a key step for evaluating $\underset{\alpha_2>|g|^2>\alpha_0} {\mathcal{E}}\left\{ S_m\right\}$ is to calculate the following general expectation:
\begin{align}\nonumber
&\underset{\alpha_2>|g|^2>\alpha_0} {\mathcal{E}}\left\{ e^{-p\xi}e^{-\mu|g|^2}\right\}
\\\nonumber =&e^{-p\alpha_s}\int^{\alpha_2}_{\alpha_1}e^{-(1+\mu) x}dx +\int^{\alpha_1}_{\alpha_0}e^{-p\frac{\alpha_0^{-1} x-1}{P_s}}e^{-(1+\mu) x}dx
\\  =&e^{-p\alpha_s}g_{\mu}(\alpha_1,\alpha_2)   +e^{\frac{ p}{P_s}}  g_{\mu+\frac{ p}{P_s\alpha_0}}(\alpha_0,\alpha_1)\label{define of g}. 
\end{align}

 By using the result shown in \eqref{define of g}, $ {Q}_{m}  $ can be calculated as follows:
 \begin{align}\label{amx}
 {Q}_{m} =   
   \eta_m \sum^{M-m-2}_{i=0}{M-m-2 \choose i}\frac{(-1)^i}{m(i+1)}  \sum^{m}_{p=0}{m \choose p}(-1)^p  \\\nonumber \times   
  \left(\frac{\mu_3  \phi(p,\mu_1) -\mu_5  \phi(p,\mu_6)   }{M-m}  - \frac{\mu_4  \phi(p,\mu_2)-\mu_5  \phi(p,\mu_6) }{M-m-1-i}\right) .
\end{align}
  
The form in  \eqref{amx} is   quite involved and cannot be directly used to obtain a high SNR approximation later. In the following, we will show that the expression of $Q_m$ can be simplified. In particular, $Q_m$ can be first rewritten as follows: 
   \begin{align} 
 {Q}_{m} =&   
  \frac{ \eta_m}{m(M-m-1)} \sum^{M-m-2}_{i=0}{M-m-1 \choose i+1}(-1)^i  \sum^{m}_{p=0}{m \choose p}(-1)^p  \\\nonumber &\times   
  \left(\frac{\mu_3  \phi(p,\mu_1) -\mu_5  \phi(p,\mu_6)   }{M-m}  - \frac{\mu_4  \phi(p,\mu_2)-\mu_5  \phi(p,\mu_6) }{M-m-1-i}\right)  ,
\end{align}
which is obtained by absorbing $i+1$ into the binomial coefficients ${M-m-1 \choose i+1}$. By letting $l=i+1$, $Q_m$ can be further rewritten as follows:
   \begin{align} 
 {Q}_{m} =&      
  \frac{ -\eta_m}{m(M-m-1)} \sum^{M-m-1}_{l=0}{M-m-1 \choose l}(-1)^l  \sum^{m}_{p=0}{m \choose p}(-1)^p  \\\nonumber &\times   
  \left(\frac{\mu_3  \phi(p,\mu_1) -\mu_5  \phi(p,\mu_6)   }{M-m}  - \frac{\tilde{\mu}_4  \phi(p,\tilde{\mu}_2)-\mu_5  \phi(p,\mu_6) }{M-m-l}\right) .
\end{align} 
We note that the term $l=0$ can be added since $\frac{\mu_3  \phi(p,\mu_1) -\mu_5  \phi(p,\mu_6)   }{M-m}  - \frac{\tilde{\mu}_4  \phi(p,\tilde{\mu}_2)-\mu_5  \phi(p,\mu_6) }{M-m-l}=0$ for $l=0$.

We further note the fact that $\mu_1$, $\mu_3$, $\mu_5$ and $\mu_6$ are not functions of $l$. Therefore,   by using the fact that $ \sum^{n}_{l=0}(-1)^l  {n \choose l}=0$, some terms in $Q_m$ can be eliminated. In particular,   $Q_m$ can be simplified as follows:
   \begin{align} \nonumber
 {Q}_{m} =&       
 \frac{ \eta_m}{m(M-m-1)} \sum^{M-m-1}_{l=0}{M-m-1 \choose l}(-1)^l    \sum^{m}_{p=0}{m \choose p}(-1)^p    
   \frac{\tilde{\mu}_4  \phi(p,\tilde{\mu}_2)-\mu_5  \phi(p,\mu_6) }{M-m-l}\\\nonumber
 \overset{(a)}{  =}&       
\bar{\eta}_m  \sum^{M-m}_{l=0}{M-m \choose l}(-1)^l  \sum^{m}_{p=0}{m \choose p}(-1)^p      
 \left(\tilde{\mu}_4  \phi(p,\tilde{\mu}_2)-\mu_5  \phi(p,\mu_6) \right) 
 \\ \label{am}
    \overset{(b)}{  =}&       
\bar{\eta}_m  \sum^{M-m}_{l=0}{M-m \choose l}(-1)^l  \sum^{m}_{p=0}{m \choose p}(-1)^p   \tilde{\mu}_4  \phi(p,\tilde{\mu}_2)  ,
\end{align}
where step (a) follows by absorbing  $M-m-l$ into the binomial coefficients, step (b) follows by using the fact that $ \sum^{n}_{l=0}(-1)^l  {n \choose l}=0$.  Again, we note that in step (a), the term $l=M-m$ can be added without changing  the value of the summation since $\tilde{\mu}_4  \phi(p,\tilde{\mu}_2)-\mu_5  \phi(p,\mu_6) =0$ for $l=M-m$.  Comparing \eqref{amx} to \eqref{am}, we note that the expression for $Q_m$ has been simplified.

\subsection{Evaluation of $Q_{M-1}$}
 Recall that  $  {Q}_{M-1} $ can be expressed as follows: 
 \begin{align}  \label{qm-1}
  {Q}_{M-1} &=\underset{|g|^2>\alpha_0} {\mathcal{E}}\left\{{\rm P}\left( |h_{M-1}|^2 <\xi,\right.\right.
\\\nonumber &\left.\left. |h_{M}|^2 >\frac{\alpha_0^{-1}|g|^2-1}{P_s} , |h_M|^2<\alpha_s(1+P_0|g|^2)  \right)\right\}.
\end{align}

Denote  the probability inside of the expectation in \eqref{qm-1} by $S_{M-1}$. Again, by applying \eqref{range},   ${Q}_{M-1} $ can be expressed as follows:
\begin{align} 
  {Q}_{M-1} = &\underset{\alpha_2>|g|^2>\alpha_0} {\mathcal{E}}\left\{ S_{M-1}\right\}.
\end{align}

Unlike $S_m$, $1\leq m \leq M-2$,  $S_{M-1}$ becomes a function of two order statistics,  $|h_{M-1}|^2$ and $|h_{M}|^2$, whose joint pdf is given by 
\begin{align}
&f_{|h_{M-1}|^2,  |h_{M}|^2}(x,y) = \tilde{\eta}_0 e^{-x} \left(1-e^{-x}\right)^{M-2} e^{-y},
\end{align}
where $x\leq y$.  By using this joint pdf, $S_{M-1}$ can be calculated as follows: 
\begin{align}\nonumber
S_{M-1} = & \frac{\tilde{\eta}_0\left(1-e^{-\xi}\right)^{M-1} \left(
 e^{-\frac{\alpha_0^{-1}|g|^2-1}{P_s} } - e^{-\alpha_s(1+P_0|g|^2)} 
 \right)}{M-1}\\\nonumber  = & \sum^{M-1}_{i=0}{M-1 \choose i} (-1)^i \frac{\tilde{\eta}_0}{M-1} e^{-i\xi} \\&\times \left(
  e^{\frac{ 1}{P_s} } e^{-\mu_7 |g|^2} - e^{-\alpha_s } 
 e^{-\mu_8|g|^2} 
 \right). 
\end{align}  

By applying \eqref{define of g}, $ {Q}_{M-1} $ can be obtained as follows:
 \begin{align} \label{am-1}
 {Q}_{M-1} 
 = & \sum^{M-1}_{i=0}{M-1 \choose i} (-1)^i \frac{\tilde{\eta}_0}{M-1}  \\&\times \left(
  e^{\frac{ 1}{P_s} } \phi(i,\mu_7)   - e^{-\alpha_s } 
\phi(i,\mu_8)
 \right).
\end{align}

 \subsection{Evaluation of  $Q_M$}
 Unlike $Q_m$, $1\leq m \leq M-1$, $Q_M$ is a function of $h_M$ and $g$. In particular, recall that $Q_M$ can be expressed as follows:
 \begin{align}
Q_M= &{\rm P}\left(\log\left(1+P_s|h_M|^2\right)<R_s,\right.\\\nonumber &\left.\log\left(1+\frac{P_s|h_M|^2}{P_0|g|^2+1}\right) <R_s, E_M,|g|^2>\alpha_0  \right)\\\nonumber
= &{\rm P}\left(\log\left(1+P_s|h_M|^2\right)<R_s ,|h_{M}|^2<\tau(|g|^2) P_s^{-1}, |g|^2>\alpha_0  \right),
\end{align}
where the last step follows from the fact that $P_s|h_M|^2\geq \frac{P_s|h_M|^2}{P_0|g|^2+1}$.  Therefore, we can rewrite $Q_m$
as follows:
 \begin{align}\nonumber 
Q_M=  &{\rm P}\left( |h_M|^2<\alpha_s ,|h_{M}|^2<\frac{\alpha_0^{-1}|g|^2 -1}{P_s},|g|^2>\alpha_0 \right)
\\ 
=  &{\rm P}\left( |h_{M}|^2<\frac{\alpha_0^{-1}|g|^2 -1}{P_s},\alpha_0 <|g|^2<\alpha_1  \right)
+{\rm P}\left( |h_M|^2<\alpha_s  ,|g|^2>\alpha_1   \right),
\end{align}
where we use the fact that $\alpha_s  <\frac{\alpha_0^{-1}|g|^2 -1}{P_s}$ is guaranteed if $|g|^2>\alpha_1$. By applying the fact that $h_M$ and $g$ are independent, 
 $Q_M$ can be calculated as follows:
  \begin{align}
Q_M=   &\int^{\alpha_1 }_{\alpha_0 }\left(1-e^{-\frac{\alpha_0^{-1}x -1}{P_s}}\right)^M e^{-x}dx
\\\nonumber
 &+\left(1-e^{-  \alpha_s} \right) ^M  e^{- \alpha_1 } .
\end{align}
With some algebraic manipulations,  $Q_M$ can be finally  expressed as follows:
  \begin{align}
Q_M
 = &\sum^{M}_{i=0}{M\choose i} (-1)^i e^{\frac{ i}{P_s}}g_{\frac{i }{\alpha_0P_s}}\left(\alpha_0,\alpha_1\right)
\\\nonumber
 &+\left(1-e^{-  \alpha_s} \right) ^M  e^{- \alpha_1 } .\label{aM}
\end{align}

\subsection{Evaluation of  $Q_0$}
$Q_0$ is surprisingly more complicated  to analyze, compared to $Q_M$. Recall that   $Q_0$ can be expressed as follows:
 \begin{align}
Q_0=& {\rm P}\left(\log\left(1+\frac{P_s|h_M|^2}{P_0|g|^2+1}\right)<R_s ,  |h_{1}|^2>\frac{P_0\epsilon_0^{-1}|g|^2-1}{P_s} ,|g|^2>\frac{\epsilon_0}{P_0}\right)
\\\nonumber 
=& {\rm P}\left( |h_M|^2<\alpha_s(P_0|g|^2+1) ,   |h_{1}|^2>\frac{\alpha_0^{-1}|g|^2-1}{P_s} ,|g|^2>\alpha_0\right).
\end{align}
Again, by applying the fact that the lower bound on $|h_M|^2$ needs to be larger than the upper bound on $|h_1|^2$ as discussed in  \eqref{range}, the probability $Q_0$ can be expressed as follows: 
 \begin{align}\nonumber 
Q_0=&  \underset{\alpha_0<|g|^2<\alpha_2}{\mathcal{E}} \left\{{\rm P}\left( |h_M|^2<\alpha_s(P_0|g|^2+1)  ,  |h_{1}|^2>\frac{\alpha_0^{-1}|g|^2-1}{P_s}  \right)\right\}. \label{q00}
\end{align}
Denote the probability inside   the expectation in \eqref{q00} by $S_0$. $S_0$ is a function of two order statistics, $|h_1|^2$ and $|h_M|^2$, whose joint pdf is given by
\begin{align}
f_{|h_1|^2,|h_M|^2}(x,y)=&
\tilde{\eta}_0 e^{-x}\left(e^{-x}-e^{-y}\right)^{M-2}e^{-y} \\\nonumber= & \tilde{\eta}_0 \sum^{M-2}_{i=0}(-1)^i {M-2 \choose i}e^{-(M-1-i)x} e^{-(i+1)y} ,
\end{align}
for $x\leq y$.  
For a fixed $|g|^2$ and by applying the joint pdf,  $S_0$ can be calculated as follows: 
 \begin{align}
S_0&=  \tilde{\eta}_0 \sum^{M-2}_{i=0}(-1)^i   {M-2 \choose i}\\\nonumber &\times \int^{\alpha_s(P_0|g|^2+1)  }_{\frac{\alpha_0^{-1}|g|^2-1}{P_s}  }e^{-(M-1-i)x}\int^{\alpha_s(P_0|g|^2+1) }_{x}  e^{-(i+1)y} dydx. 
\end{align}
With some algebraic manipulations, $S_0$ can be expressed as follows:
 \begin{align}
S_0&=   \tilde{\eta}_0 \sum^{M-2}_{i=0}(-1)^i {M-2 \choose i} \\\nonumber &\times   
\left(\frac{
e^{\frac{ M}{P_s} }e^{-\mu_{10}|g|^2 }
 - e^{- M\alpha_s  }e^{-\mu_{11}|g|^2 }
}{M(i+1)}\right.\\\nonumber
& -\frac{ 
\left(
e^{\frac{M-1-i}{P_s}- (i+1)\alpha_s  } e^{-\mu_{12}|g|^2} - 
e^{-M \alpha_s  } e^{-\mu_{11}|g|^2 } 
\right)
}{(i+1)(M-1-i)}  ,
\end{align}
where  $\mu_{10}=\frac{M}{\alpha_0P_s}$ and $\mu_{11}= M\alpha_sP_0 $, and $\mu_{12}= (i+1)\alpha_sP_0+ (M-1-i)\frac{\alpha_0^{-1}}{P_s} $.

 By applying the integration result in \eqref{define of g},  $Q_0$ can be expressed as follows: 
 \begin{align}\nonumber
Q_{0}&=   \tilde{\eta}_0 \sum^{M-2}_{i=0}(-1)^i  {M-2 \choose i}\\\nonumber &\times   
\left(\frac{
e^{\frac{ M}{P_s} }g_{\mu_{10}}(\alpha_0,\alpha_2) 
 - e^{- M\alpha_s  }g_{\mu_{11}}(\alpha_0,\alpha_2) 
}{M(i+1)}\right.\\ 
& \left.-\frac{ 
\left(
e^{\frac{M-1-i}{P_s}- (i+1)\alpha_s  } g_{\mu_{12}}(\alpha_0,\alpha_2) - 
e^{-M \alpha_s  } g_{\mu_{11}}(\alpha_0,\alpha_2) 
\right)
}{(i+1)(M-1-i)}  \right).\label{a00}
\end{align}

The expression  in \eqref{a00} is quite involved, and cannot be used directly to obtain a high SNR approximation. In the following, we will show that \eqref{a00} can be simplified. First,  the expression for $Q_0$ is modified as follows:
\begin{align}
Q_{0}=&   \frac{\tilde{\eta}_0}{M-1} \sum^{M-2}_{i=0}(-1)^i  {M-1 \choose i+1}\\\nonumber &\times   
\left(\frac{
e^{\frac{ M}{P_s} }g_{\mu_{10}}(\alpha_0,\alpha_2) 
 - e^{- M\alpha_s  }g_{\mu_{11}}(\alpha_0,\alpha_2) 
}{M}\right.\\\nonumber
& \left.-\frac{ 
\left(
e^{\frac{M-1-i}{P_s}- (i+1)\alpha_s  } g_{\mu_{12}}(\alpha_0,\alpha_2) - 
e^{-M \alpha_s  } g_{\mu_{11}}(\alpha_0,\alpha_2) 
\right)
}{(M-1-i)}  \right) ,
\end{align}
which is obtained by absorbing $i+1$ into the binomial coefficients. The expression for $Q_0$ can be further modified as follows:
\begin{align}\label{high x1}
Q_{0}=&    -  \frac{\tilde{\eta}_0}{M-1} \sum^{M-1}_{l=0}(-1)^l  {M-1 \choose l}\\\nonumber &\times   
\left(\frac{
e^{\frac{ M}{P_s} }g_{\mu_{10}}(\alpha_0,\alpha_2) 
 - e^{- M\alpha_s  }g_{\mu_{11}}(\alpha_0,\alpha_2) 
}{M }\right.\\\nonumber
& \left.-e^{ - l\alpha_s  }\frac{ 
\left(
e^{\frac{M-l}{P_s}   } g_{\tilde{\mu}_{12}}(\alpha_0,\alpha_2) - 
e^{-(M-l) \alpha_s  } g_{\mu_{11}}(\alpha_0,\alpha_2) 
\right)
}{(M-l)}  \right)
\end{align}
which is obtained by substituting  $l=i+1$. We note that  the term $l=0$ can be added without changing the value of the summation, since  the terms in the second and third lines in \eqref{high x1} cancel each other when $l=0$.  

By using the fact that $ \sum^{n}_{l=0}(-1)^l  {n \choose l}=0$, $Q_0$ can be further simplified as follows: 
\begin{align}
Q_{0} 
=&   \frac{\tilde{\eta}_0}{M-1} \sum^{M-1}_{l=0}(-1)^l  {M-1 \choose l}e^{ - l\alpha_s  }\\\nonumber &\times   
\left( \frac{ 
e^{\frac{M-l}{P_s}   } g_{\tilde{\mu}_{12}}(\alpha_0,\alpha_2) - 
e^{-(M-l) \alpha_s  } g_{\mu_{11}}(\alpha_0,\alpha_2) 
}{(M-l)}  \right)
\\\nonumber
=&  \frac{ \tilde{\eta}_0}{M(M-1)} \sum^{M}_{l=0}(-1)^l  {M \choose l}e^{ - l\alpha_s  }\\\nonumber &\times   
\left(  
e^{\frac{M-l}{P_s}   } g_{\tilde{\mu}_{12}}(\alpha_0,\alpha_2) - 
e^{-(M-l) \alpha_s  } g_{\mu_{11}}(\alpha_0,\alpha_2) 
 \right),
\end{align}
where the last step is obtained by absorbing $(M-l)$ into the binomial coefficients.  In addition, we also note that the term $l=M$ can be added since $
e^{\frac{M-l}{P_s}   } g_{\tilde{\mu}_{12}}(\alpha_0,\alpha_2) - 
e^{-(M-l) \alpha_s  } g_{\mu_{11}}(\alpha_0,\alpha_2) =0$ when $l=M$. 
 
Again, by using the fact that $ \sum^{n}_{l=0}(-1)^l  {n \choose l}=0$, $Q_0$ can be further simplified as follows: 
\begin{align}\label{a0}
Q_{0} 
=&    \frac{ \tilde{\eta}_0}{M(M-1)} \sum^{M}_{l=0}(-1)^l  {M \choose l}e^{ - l\alpha_s  }  
e^{\frac{M-l}{P_s}   } g_{\tilde{\mu}_{12}}(\alpha_0,\alpha_2)   .  
\end{align}

$Q_{M+1}$ can be evaluated similarly to $Q_M$ since both are functions of $h_M$ and $g$, and it can be expressed  as follows: 
\begin{align}\nonumber
Q_{M+1}=&{\rm P}\left(\log\left(1+\frac{P_s|h_M|^2}{P_0|g|^2+1}\right)<R_s , |g|^2<\alpha_0\right) \\ =&\sum^{M}_{i=0}{M\choose i}(-1)^ie^{- i\alpha_s }
 \frac{1-e^{-\left(1+ i\alpha_sP_0 \right)\alpha_0 } }{1+ i\alpha_sP_0 }. \label{aM+1}
\end{align}

Therefore, by combining \eqref{am}, \eqref{am-1}, \eqref{aM}, \eqref{a0} and \eqref{aM+1}, the overall outage probability is obtained as shown in the theorem and the proof is complete.

\section{Proof for Theorem \ref{theorem2}}
As discussed in the proof for Theorem $\ref{theorem1}$,    $Q_m$ depends on the value of   $m$. Therefore, the high SNR approximations for different   $Q_m$ will be discussed separately in the following subsections.

\subsection{High SNR Approximation for  $Q_m$, $1\leq m \leq M-2$}
 Among all the terms in \eqref{overall}, the expression for $Q_m$, $1\leq m\leq M-2$, is the most complicated one, as is evident from the proof of Theorem \ref{theorem1}. First, recall  $ {Q}_{m} $ can be   expressed as follows:
  \begin{align} \label{qm simplified}
 {Q}_{m} =&         
\bar{\eta}_m  \sum^{M-m}_{l=0}{M-m \choose l}(-1)^l  \sum^{m}_{p=0}{m \choose p}(-1)^p   \tilde{\mu}_4  \phi(p,\tilde{\mu}_2)  \\\nonumber =&
\bar{\eta}_m  \sum^{M-m}_{l=0}{M-m \choose l}(-1)^l  \sum^{m}_{p=0}{m \choose p}(-1)^p  \\\nonumber &\times   
 \left(\tilde{\mu}_4 e^{-p\alpha_s}g_{\tilde{\mu}_2}(\alpha_1,\alpha_2)   +\tilde{\mu}_4 e^{\frac{ p}{P_s}}  g_{\tilde{\mu}_2+\frac{ p}{P_s\alpha_0}}(\alpha_0,\alpha_1)   \right) .
\end{align}

In order to facilitate a high SNR approximation,  $ {Q}_{m}$ is rewritten as follows:
  \begin{align} 
 {Q}_{m} =&    
\bar{\eta}_m \int^{\alpha_2}_{\alpha_1} \sum^{M-m}_{l=0}{M-m \choose l}(-1)^l  \sum^{m}_{p=0}{m \choose p}(-1)^p    
 \tilde{\mu}_4 e^{-p\alpha_s}e^{-(1+\tilde{\mu}_2)x}   dx   \\\nonumber&+ 
\bar{\eta}_m \int^{\alpha_1}_{\alpha_0} \sum^{M-m}_{l=0}{M-m \choose l}(-1)^l  \sum^{m}_{p=0}{m \choose p}(-1)^p  
 \tilde{\mu}_4 e^{\frac{ p}{P_s}}  e^{-(1+\tilde{\mu}_2+\frac{ p}{P_s\alpha_0})x}dx ,
\end{align}

By applying the   approximation, $e^{-x}\approx 1-x$ for $x\rightarrow 0$ and also using the definitions of $\tilde{\mu}_2$ and $\tilde{\mu}_4$, $Q_m$ can be approximated as follows:
 \begin{align} 
 {Q}_{m}  
 \approx&    
\bar{\eta}_m \int^{\alpha_2}_{\alpha_1} \sum^{M-m}_{l=0}{M-m \choose l}(-1)^l  \sum^{m}_{p=0}{m \choose p}(-1)^p  \\\nonumber &\times   
e^{-l\alpha_s-\frac{ l}{P_s} }e^{-p\alpha_s}e^{-( l\epsilon_s-l \epsilon_0^{-1})x}   dx   \\\nonumber&+ 
\bar{\eta}_m \int^{\alpha_1}_{\alpha_0} \sum^{M-m}_{l=0}{M-m \choose l}(-1)^l  \sum^{m}_{p=0}{m \choose p}(-1)^p  \\\nonumber &\times   
e^{-l\alpha_s-l\frac{ 1}{P_s} } e^{\frac{ p}{P_s}}  e^{-( l\epsilon_s-l \epsilon_0^{-1}+\frac{ p}{P_s\alpha_0})x}dx.
\end{align}

By    using the fact that $ \sum^{n}_{l=0}(-1)^l  {n \choose l}a^l=(1-a)^n$, the approximation of  $Q_m$ can be further simplified  as follows:
  \begin{align} 
 {Q}_{m}  
 \approx&    
\bar{\eta}_m\left(1-e^{-\alpha_s}\right)^{m} \int^{\alpha_2}_{\alpha_1} \sum^{M-m}_{l=0}{M-m \choose l}(-1)^l    
e^{-l\alpha_s-\frac{ l}{P_s} }e^{-( l\epsilon_s-l \epsilon_0^{-1})x}   dx   \\\nonumber&+ 
\bar{\eta}_m  \int^{\alpha_1}_{\alpha_0} \left(1-e^{\frac{ 1}{P_s}-\frac{ x}{P_s\alpha_0}} \right)^m\sum^{M-m}_{l=0}{M-m \choose l}(-1)^l    
e^{-l\alpha_s-l\frac{ 1}{P_s} }  e^{-( l\epsilon_s-l \epsilon_0^{-1})x}dx
\\\nonumber
 =&    
\bar{\eta}_m\left(1-e^{-\alpha_s}\right)^{m} \int^{\alpha_2}_{\alpha_1}      
\left(1-e^{-(\alpha_s+\frac{ 1}{P_s} +( \epsilon_s- \epsilon_0^{-1})x)} \right)^{M-m} dx   \\\nonumber&+ 
\bar{\eta}_m \int^{\alpha_1}_{\alpha_0}\left(1-e^{\frac{ 1}{P_s}-\frac{ x}{P_s\alpha_0}} \right)^m 
\left(1-e^{-(\alpha_s+\frac{ 1}{P_s} +( \epsilon_s- \epsilon_0^{-1} )x}\right)^{M-m}dx.
\end{align}

A more simplified  form of $ {Q}_{m} $ can be obtained by carrying out the following high SNR approximations: 
  \begin{align} 
 {Q}_{m}  
 \approx&     
\bar{\eta}_m\left(1-e^{-\alpha_s}\right)^{m} \int^{\alpha_2}_{\alpha_1}       
\left( \alpha_s+\frac{ 1}{P_s} +( \epsilon_s- \epsilon_0^{-1})x  \right)^{M-m} dx   \\\nonumber&+ 
\frac{\bar{\eta}_m}{P_s^m\alpha_0^m} \int^{\alpha_1}_{\alpha_0}\left( x-\alpha_0 \right)^m    \left( \alpha_s+\frac{ 1}{P_s} +( \epsilon_s- \epsilon_0^{-1}  )x\right)^{M-m}dx ,
\end{align}
 With some algebraic manipulations,  the high SNR approximation for $Q_m$ can be obtained as follows:
  \begin{align}\nonumber 
 {Q}_{m}  
 \approx       &
\frac{\bar{\eta}_m}{P_s^{M+1}}\epsilon_s^{m}\sum^{M-m}_{i=0} {M-m \choose i}  \left( \epsilon_s+1  \right)^{M-m-i}  \left(  \epsilon_s- \epsilon_0^{-1}   \right)^{i}  \\\nonumber &\times   
\epsilon_0^{i+1}\frac{\tilde{\alpha}_2^{i+1}-(1+\epsilon_s)^{i+1}}{i+1}   + 
\frac{\bar{\eta}_m}{P_s^{M+1}\epsilon_0^m} \sum^{M-m}_{i=0} {M-m \choose i}  \\ \label{aam} &\times   \left( \epsilon_s+\epsilon_0 \epsilon_s  \right)^{M-m-i}  
\left(   \epsilon_s- \epsilon_0^{-1}  \right)^{i}\epsilon_0^{m+i+1}\frac{\epsilon_s^{m+i+1}}{m+i+1}  .
\end{align}

\subsection{High SNR Approximation for $Q_0$}
To facilitate the  asymptotic analysis,  $Q_0$ can be   rewritten  as follows: 
\begin{align}\nonumber
Q_{0} 
=&    \frac{ \tilde{\eta}_0}{M(M-1)} \sum^{M}_{l=0}(-1)^l  {M \choose l}e^{ - l\alpha_s  }  
e^{\frac{M-l}{P_s}   } g_{\tilde{\mu}_{12}}(\alpha_0,\alpha_2)  
\\\nonumber
=&    \frac{ \tilde{\eta}_0}{M(M-1)} \sum^{M}_{l=0}(-1)^l  {M \choose l}e^{ - l\alpha_s  }  
e^{\frac{M-l}{P_s}   }  \frac{e^{-(1+\tilde{\mu}_{12}) \alpha_0}-e^{-(1+\tilde{\mu}_{12}) \alpha_2}}{1+\tilde{\mu}_{12}} .  
\end{align}

To carry out high SNR approximations,  $Q_{0} $ can be first rewritten   as follows:
  \begin{align}
Q_{0} 
=&       \frac{ \tilde{\eta}_0}{M(M-1)} \sum^{M}_{l=0}(-1)^l  {M \choose l}e^{ - l\alpha_s  }  
e^{\frac{M-l}{P_s}   }    \int^{\alpha_2}_{\alpha_0} e^{-(1+\tilde{\mu}_{12})x}dx 
\\\nonumber=&        \frac{ \tilde{\eta}_0}{M(M-1)}e^{\frac{M}{P_s}   } \int^{\alpha_2}_{\alpha_0}e^{-x} e^{- M\epsilon_0^{-1} x}\sum^{M}_{l=0}(-1)^l  {M \choose l}
 \\\nonumber
&\times  e^{ - l\left(\alpha_s + \frac{1}{P_s}+\epsilon_sx -\epsilon_0^{-1} x\right)}    dx  .
\end{align}
By using the fact that $ \sum^{n}_{l=0}(-1)^l  {n \choose l}a^l=(1-a)^n$, $Q_0$ can be expressed as follows:
  \begin{align}
Q_{0} 
=&    
       \frac{ \tilde{\eta}_0}{M(M-1)}e^{\frac{M}{P_s}   } \int^{\alpha_2}_{\alpha_0} 
   \left(1-e^{ - l\left(\alpha_s + \frac{1}{P_s}+\epsilon_sx -\epsilon_0^{-1} x\right)}   \right)^M dx 
\\\nonumber\approx& 
       \frac{ \tilde{\eta}_0}{M(M-1)}e^{\frac{M}{P_s}   } \int^{\alpha_2}_{\alpha_0}   \left( \alpha_s + \frac{1}{P_s}+\epsilon_sx -\epsilon_0^{-1} x  \right)^M dx ,
\end{align}
where the last step is obtained by applying the following approximation, $e^{-x}\approx 1-x$ for $x\rightarrow 0$. By applying the binomial expansion,  $Q_0$ can be expressed as follows:
 \begin{align}
Q_{0} \nonumber
\approx&     
       \frac{ \tilde{\eta}_0}{M(M-1)}e^{\frac{M}{P_s}   } \sum^{M}_{i=0}{M\choose i}\left(\alpha_s + \frac{1}{P_s} \right)^{M-i}
      \left( \epsilon_s -\epsilon_0^{-1} \right)^i  \int^{\alpha_2}_{\alpha_0}x^i     dx 
\\\nonumber
=&    
       \frac{ \tilde{\eta}_0}{M(M-1)}e^{\frac{M}{P_s}   } \sum^{M}_{i=0}{M\choose i}\left(\alpha_s + \frac{1}{P_s} \right)^{M-i}
 \\ \label{aa0}
&\times     \left( \epsilon_s -\epsilon_0^{-1} \right)^i     \frac{\alpha_2^{i+1}-\alpha_0^{i+1}}{i+1}.  
\end{align}

\subsection{High SNR Approximation for $Q_{M-1}$}
First, we recall that $Q_{M-1}$ can be expressed as follows:
 \begin{align} 
{Q}_{M-1}    =  &\sum^{M-1}_{i=0}{M-1 \choose i} (-1)^i \frac{\tilde{\eta}_0}{M-1}  \\\nonumber&\times \left(
  e^{\frac{ 1}{P_s} }  e^{-i\alpha_s}g_{\mu_7}(\alpha_1,\alpha_2)   +e^{\frac{ 1}{P_s} } e^{\frac{ i}{P_s}}  g_{\mu_7+\frac{ i}{P_s\alpha_0}}(\alpha_0,\alpha_1)\right.
  \\\nonumber &\left. - e^{-\alpha_s } 
  e^{-i\alpha_s}g_{\mu_8}(\alpha_1,\alpha_2)   -e^{-\alpha_s } e^{\frac{ i}{P_s}}  g_{\mu_8+\frac{ i}{P_s\alpha_0}}(\alpha_0,\alpha_1)
 \right).
\end{align} 
 
 In order to obtain  the high SNR approximation, we first express $Q_{M-1}$ as follows: 
 \begin{align} 
 {Q}_{M-1}  
  = & \sum^{M-1}_{i=0}{M-1 \choose i} (-1)^i \frac{\tilde{\eta}_0}{M-1}  \\&\times \left(
\int^{\alpha_2}_{\alpha_1}  e^{\frac{ 1}{P_s} }  e^{-i\alpha_s}e^{-(1+ \mu_7)x} dx    +\int^{\alpha_1}_{\alpha_0} e^{\frac{ 1}{P_s} } e^{\frac{ i}{P_s}}  e^{-(1+ \mu_7+\frac{ i}{P_s\alpha_0})x}dx
  \right.
  \\\nonumber &\left. - \int^{\alpha_2}_{\alpha_1}e^{-\alpha_s } 
  e^{-i\alpha_s}e^{-(1+ \mu_8)x} dx   -\int^{\alpha_1}_{\alpha_0}e^{-\alpha_s } e^{\frac{ i}{P_s}}  e^{-(1+ \mu_8+\frac{ i}{P_s\alpha_0})x}dx
 \right) .\end{align}

 By using the fact that $ \sum^{n}_{l=0}(-1)^l  {n \choose l}=0$, $ {Q}_{M-1}   $ can be further expressed as follows:
 \begin{align} \nonumber
& {Q}_{M-1}   
 =    \frac{\tilde{\eta}_0}{M-1}   \left(e^{\frac{ 1}{P_s} }  \left(1-e^{-\alpha_s}\right)^{M-1}
\int^{\alpha_2}_{\alpha_1}  e^{-(1+ \frac{1}{P_s\alpha_0})x} dx  \right.\\\nonumber
  &\left. +\int^{\alpha_1}_{\alpha_0} e^{\frac{ 1}{P_s} } \left(1-e^{\frac{ 1}{P_s}-\frac{ 1}{P_s\alpha_0}x}\right)^{M-1}  e^{-(1+ \frac{1}{P_s\alpha_0})x}dx
  \right.
  \\\nonumber &\left. -e^{-\alpha_s } 
  \left(1-e^{-\alpha_s}\right)^{M-1} \int^{\alpha_2}_{\alpha_1}e^{-(1+ \alpha_sP_0)x} dx\right. \\  &\left.  -\int^{\alpha_1}_{\alpha_0}e^{-\alpha_s } \left(1-e^{\frac{ 1}{P_s}-\frac{ 1}{P_s\alpha_0}x}  \right)^{M-1}e^{-(1+ \alpha_sP_0)x}dx
 \right).
\end{align} 

Directly applying the approximation, $e^{-x}\approx 1-x$ for $x\rightarrow 0$, to the above equation results in a very complicated form. In order to facilitate the high SNR approximation, we rearrange the four terms in the above equation as follows:
  \begin{align} \nonumber
 {Q}_{M-1}   
 = &   \frac{\tilde{\eta}_0}{M-1}   \left(  \left(1-e^{-\alpha_s}\right)^{M-1}\right.\\\nonumber&\times
 \left(
\int^{\alpha_2}_{\alpha_1}  e^{\frac{ 1}{P_s}-(1+ \frac{1}{P_s\alpha_0})x} -e^{-\alpha_s -(1+ \alpha_sP_0)x} dx\right)   \\\nonumber
  &\left. +\int^{\alpha_1}_{\alpha_0}  \left(1-e^{\frac{ 1}{P_s}-\frac{ 1}{P_s\alpha_0}x}\right)^{M-1}  
  \right.
  \\ \nonumber &\left.  \times\left(e^{\frac{ 1}{P_s}-(1+ \frac{1}{P_s\alpha_0})x} -   e^{-\alpha_s-(1+ \alpha_sP_0)x}\right)dx
 \right) .
\end{align} 
By applying  the approximation, $e^{-x_1}-e^{-x_2}\approx x_2-x_1$ for $x_1\rightarrow 0$ and $x_2\rightarrow 0$, $ {Q}_{M-1}  $ can be approximated as follows:
 \begin{align}  
 {Q}_{M-1}    \label{qm-dd}
  \approx&   \frac{\tilde{\eta}_0}{M-1}   \left(  \left(1-e^{-\alpha_s}\right)^{M-1}\right.\\\nonumber&\times
 \left(
\int^{\alpha_2}_{\alpha_1}  \left(\frac{ 1}{P_s}- \frac{1}{P_s\alpha_0}x +\alpha_s + \alpha_sP_0x\right) dx\right)   \\\nonumber
  &\left. +\int^{\alpha_1}_{\alpha_0}  \left(1-e^{\frac{ 1}{P_s}-\frac{ 1}{P_s\alpha_0}x}\right)^{M-1}  
  \right.
  \\ \nonumber &\left.  \times\left(\frac{ 1}{P_s}- \frac{1}{P_s\alpha_0}x+\alpha_s+ \alpha_sP_0x\right)dx
 \right).
\end{align}

The approximation shown in \eqref{qm-dd} can be further approximated as follows: 
 \begin{align} \nonumber
 {Q}_{M-1}    
  \approx&   \frac{\tilde{\eta}_0}{M-1}   \left(  \left(1-e^{-\alpha_s}\right)^{M-1} 
 \left(
\int^{\alpha_2}_{\alpha_1}  \left(\frac{ 1}{P_s}  +\alpha_s  \right) dx\right) \right.  \\\nonumber
  &\left. +\int^{\alpha_1}_{\alpha_0}  \left(1-e^{\frac{ 1}{P_s}-\frac{ 1}{P_s\alpha_0}x}\right)^{M-1}  
  \times\left(\frac{ 1}{P_s} +\alpha_s\right)dx
 \right)\\ \label{aaM-1}    =&   \frac{\tilde{\eta}_0 \left(1+\epsilon_s\right)\epsilon_s ^{M-1} 
 \epsilon_0}{P_s^{M+1}(M-1)}    \left(  (\tilde{\alpha}_2 -1-\epsilon_s )  +\frac{\epsilon_s}{M}
 \right).
\end{align} 
 
 Following similar  steps as for the approximation of $Q_m$, $1\leq m \leq M-1$,  $Q_M$ can be approximated as follows:
  \begin{align}\label{aaM}
Q_M   \approx &\frac{1}{(M+1)\epsilon_0^M}   \alpha_0^{M+1}  \epsilon_s^{M+1}
+   \alpha_s  ^M  ,
\end{align} 
and $Q_{M+1}$ can be approximated as follows:
\begin{align}\label{aaM+1}
Q_{M+1}&   \alpha_s ^M \frac{(1+\epsilon_0)^{M+1}-1}{P_0(M+1)}. 
\end{align}
 
 By combining \eqref{aam}, \eqref{aa0},\eqref{aaM-1}, \eqref{aaM} and \eqref{aaM+1}, the high SNR approximation for ${\rm P}_{out}$ can be obtained and the proof for Theorem \ref{theorem2} is complete. 
    \bibliographystyle{IEEEtran}
\bibliography{IEEEfull,trasfer}
 
   \end{document}